\documentclass[conference]{IEEEtran}
\usepackage{algorithm2e}
\usepackage{cite}
\usepackage{graphicx}
\usepackage{psfrag}
\usepackage{subfigure}
\usepackage{url}
\usepackage{amsmath}
\usepackage{array}
\usepackage{amssymb}
\usepackage{amsfonts}

\newtheorem{lemma}{Lemma}

\newtheorem{corollary}{Corollary}

\newtheorem{definition}{Definition}
\newtheorem{theorem}{Theorem}
\newtheorem{example}{Example}

\title{Maximum Rate of 3- and 4-Real-Symbol ML Decodable Unitary Weight STBCs}
\author{
\authorblockN{Teja~Damodaram~Bavirisetti and B.~Sundar Rajan\\}
\IEEEauthorblockA{\small{Dept. of ECE, IISc, Bangalore 560012, India, Email: \{tdamodar,bsrajan\}@ece.iisc.ernet.in}}
}
\date{\today}
\begin{document}
\maketitle
\thispagestyle{empty}	

\begin{abstract}
It has been shown recently that the maximum rate of a 2-real-symbol (single-complex-symbol) maximum likelihood (ML) decodable, square space-time block codes (STBCs) with unitary weight matrices is $\frac{2a}{2^a}$ complex symbols per channel use (cspcu) for $2^a$ number of transmit antennas \cite{KSR}. These STBCs are obtained from Unitary Weight Designs (UWDs). In this paper, we show that the maximum rates for 3- and 4-real-symbol (2-complex-symbol) ML decodable square STBCs from UWDs, for $2^{{a}}$  transmit antennas, are $\frac{3(a-1)}{2^{a}}$ and $\frac{4(a-1)}{2^{a}}$ cspcu, respectively. STBCs achieving this maximum rate are constructed. A set of sufficient conditions on the signal set, required  for these codes to achieve full-diversity are derived along with expressions for their coding gain.
\end{abstract}

\section{Introduction}
Consider an $N$ transmit antenna, $N_{r}$ receive antenna quasi-static
Rayleigh flat fading MIMO channel given by
\begin{equation}\label{channel}
    Y = XH +W
\end{equation}
where $H \in\mathbb{C}^{N \times N_{r}}$ is the channel matrix with the entries assumed to be i.i.d., circularly symmetric Gaussian random variables $\sim \mathcal{N}_\mathbb{C}\left(0,1\right)$, $X \in \mathbb{C}^{T \times N}$ is the matrix of transmitted signal, $W \in \mathbb{C}^{T \times N_{r}}$ is a complex white Gaussian noise matrix with i.i.d., entries $\sim
\mathcal{N}_{\mathbb{C}}\left(0,N_{0}\right)$ and $Y\in  \mathbb{C}^{T \times N_{r}}$ is the  matrix of received signal ($\mathbb{C}$ is the field of complex numbers). Throughout this paper, we assume $T = N.$
\begin{definition}[\textbf{LSTD}{\cite{RaR}}]
 An $N \times N$ Linear Space-Time Design
(LSTD)  or simply a design $X$ in $K$ real variables
$x_{1},$ $\ldots,$ $x_{K}$ is a matrix $    \sum_{i=1}^{K}x_{i}A_{i}$,
where $A_{i} \in \mathbb{C}^{N\times N}$, $i = 1,\ldots,K$ and $A_{1}, \ldots ,A_{K}$ are linearly independent over the field of real numbers $\mathbb{R}.$ The matrices $A_{i}$ are known as the weight matrices.
\end{definition}
\begin{definition}[\textbf{Rate}]
 The rate of an $N \times N$ design $X$ in $K$ real variables is $R = \frac{K}{2N}$ complex symbols per channel use (cspcu).
\end{definition}
\begin{definition}[\textbf{STBC}]
An $N \times N$ Space-Time Block Code (STBC) $\cal{C}$ is a finite
subset of $\mathbb{C}^{N \times N}$.
\end{definition}
An STBC can be obtained from a design $X$ by letting the vector  $(x_{1},$ $\ldots,$ $x_{K})$ take values from a finite set $ \cal{A} \subset \mathbb{R}^{\textit{K}}.$
The set $\cal{A}$ is called the signal set. Denote the STBC obtained
this way by ${\cal{C}}(X,{\cal{A}})$. Let ${\bf s}=[x_{1}~x_{2}~ \ldots ~x_{K}]^{T}$ and $S({\bf s})=\sum_{i=1}^{K}x_{i}A_{i}$. Then, we have
\begin{equation}\label{code}
{\cal{C}}({X},{\cal{A}})=\{S({\bf s})|{\bf s}\in \cal{A} \}.
\end{equation}
An STBC $\cal{C}$, whose encoding symbols ($x_{1},\cdots,x_{K}$) are chosen from a set $\cal{A}$ is said to offer \textit{full-diversity} iff for every possible codeword pair ($S,\hat{{S}}$) ($S,\hat{{S}}\in \cal{C}$), with ${{S}}\not=\hat{{S}}$, the codeword difference matrix ${{S}}-\hat{{S}}$ is full-ranked \cite{TSC}. In general, the diversity offered by a code, depends on the constellation it employs. A code can offer full-diversity for certain signal set $\cal{A}$ but not for another signal set. The CODs are special in this aspect since they offer full-diversity for any arbitrary signal set. The coding gain $\delta$ of an STBC $\cal{C}$ is defined as
\begin{equation}
\nonumber\delta={\min}_{S-\hat{S},S\not=\hat{S}}\left({\prod}_{i=1}^{r}{\lambda}_{i}\right)^{\frac{1}{r}},
\end{equation}
where ${\lambda}_{i},~i=1,2,\cdots,r$ are the non-zero eigenvalues of the matrix $\left(S-\hat{S}\right)^{H} \left(S-\hat{S}\right)$ and $r$ is the minimum of the rank of $\left({S}-\hat{{S}}\right)^{H} \left({S}-\hat{{S}}\right)$ for all possible codeword pairs
($S,\hat{{S}}$) ($S,\hat{{S}}\in \cal{C}$), with ${S}\not=\hat{{S}}$ \cite{TSC}.
\subsection{Encoding complexity and group ML Decoding}
One of the important aspects in the design of STBCs is their ML decoding complexity. This depends on their encoding complexity \cite{RaR}. If we use (\ref{code}) for encoding an STBC from a LSTD, we see that, in general, one needs to choose an element from $\cal{A}$ and then substitute for the real variables $x_{1},$ $\ldots,$ $x_{K}$ in the LSTD. This method of encoding clearly requires a look-up table (memory) with $|\cal{A}|$ entries. However, if the signal set $\cal{A}$ is a Cartesian product of $g$ smaller signal sets in $\frac{K}{g}$ real variables, then the encoding complexity can be reduced (to memory with $g|{\cal{A}}|^{\frac{1}{g}}$ entries). Moreover, if $\cal{A}=\cal{A_{\text{1}}}$$\times\cal{A}_{\text{2}}$$\times\cdots\times$${\cal{A}}_{g}$ where each $\cal{A}_{\text{i}} \subset $ ${\mathbb{R}}^{\frac{\text{K}}{g}}$ with cardinality ${|\cal{A}|}^{\frac{\text{1}}{g}}$, then the STBC $\cal{C}$ itself decomposes into a sum of $g$ different STBCs as follows.

Let $K=g\lambda$. Then, by appropriately reordering/relabeling the real variables, we can assume without loss of generality\footnote{Here we have assumed that the first $\lambda$ variables belong to first group and second $\lambda$ variables belong to second group and last $\lambda$ variables belong to the $g-$th group. In general, the partitioning of real variables into $g$-groups can be arbitrary.} that $S({\bf s})=\sum_{i=1}^{K}x_{i}A_{i}=S_{1}({\bf s}_{1})+S_{2}({\bf s}_{2})+\cdots+S_{g}({\bf s}_{g})$, where $S_{i}({\bf s}_{i}) =\sum_{j=(i-1)\lambda+1}^{i\lambda} x_{j}A_{j}$ and ${\bf s}_{i} = \lbrack x_{(i-1)\lambda+1}~ x_{(i-1)\lambda+2}~ \cdots ~x_{i\lambda} \rbrack^{T},$ for $i=1,2,\cdots,g.$ Hence, the STBC decomposes as ${\cal C} = \sum_{i=1}^g {\cal C}_i,$ where
\begin{center}
  ${\cal{C}}_1 = \{S_{1}({\bf s}_{1}) \mid {\bf s}_{1} \in {\cal{A}}_{1}\}$

  ${\cal{C}}_2 = \{S_{2}({\bf s}_{2}) \mid {\bf s}_{2} \in {\cal{A}}_{2}\}$

 \vspace{-2 mm}
$\vdots$

    ${\cal{C}}_g = \{S_{g}({\bf s}_{g}) \mid {\bf s}_{g} \in {\cal{A}}_{g}\}$.
\end{center}
For the given channel (\ref{channel}) the ML decoder is given by
\begin{equation*}
    \hat{X} = \arg  \min_{X \in \cal{C} } ||Y-XH||_{F}^{2}.
\end{equation*}
For a $g$-group encodable STBC $\cal{C}$, X=$\sum_{i=1}^{g}$X$_{i}$ for some X$_{i}\in \cal{C}_{\textit{i}}$. Let, $\text{C}_{i} = \{ A_{(i-1)\lambda+1}, ~A_{(i-1)\lambda+2}, ~\cdots, ~A_{i\lambda} \}$, where, C$_{i}$ is the set of weight matrices corresponding to STBC ${\cal{C}}_{i}$. It is shown in  \cite{KhR}-\cite{KaR2} that, the ML decoder decomposes as
\begin{equation*}
    \hat{X} = \sum_{i=1}^{g} \arg  \min_{X_{i} \in {\cal{C}}_{i} } ||Y-X_{i}H||_{F}^{2},
\end{equation*}
if the weight matrices $A_{i},i=1,\ldots,K$ satisfy the conditions
\begin{equation}\label{matrix_condn}
    {A}_{k}^{H}{A}_{l} + {A}_{l}^{H}{A}_{k} = 0 \hspace{6mm}     \forall {A}_{\textit{k}} \in {\text{C}}_{\textit{k}}, A_{\textit{l}} \in \text{C}_{\textit{l}} ,\hspace{6 mm} \textit{k}\not=\textit{l}.
\end{equation}
In other words, the component STBCs $\cal{C}_{\text{i}}$'s can then be decoded independently.
\begin{definition}[\cite{RaR2}]
A STBC ${\cal C} = \{S(s)|\text{s} \in \cal{A} \subset \mathbb{R}^{\text{K}} \}$ is said to be $g$-group decodable or $\frac{K}{g}$ real symbol decodable (or $\frac{K}{2g}$ complex symbol ML decodable) if $\cal{C}$ is $g$-group encodable and if the associated weight matrices satisfy (\ref{matrix_condn}).
\end{definition}
\subsection{Contributions}
In \cite{KSR}, an achievable upper bound on the rate of unitary-weight single-complex-symbol-decodable (SSD) code is derived to be $\frac{2a}{2^a}$ cspcu for $2^{a}$ antennas. The maximum rate of 3- and 4-real symbol ML decodable $2^{a}\times2^{a}$ ($a\ge2$) Unitary Weight Designs (UWDs) (LSTDs with unitary weight matrices) has not been reported so far in the literature, to the best of our knowledge.

The contributions  of this paper are as follows:

\begin{itemize}
\item We show that the maximum rate of 3- and 4-real symbol ML decodable $2^{a}\times2^{a}$ ($a\ge2$) UWDs are $\frac{3(a-1)}{2^a}$ cspcu and $\frac{4(a-1)}{2^a}$ cspcu, respectively. (Section \ref{sec3})
\item Codes which achieve this maximum rate are presented (Explicit construction in the proof of Theorem \ref{thmcode}).
\item For our explicitly constructed codes, signal sets achieving full-diversity have been identified  along with expressions for their coding gain (Section \ref{sec4}).
\end{itemize}
{\textit{\textbf Organization:}}  In Section \ref{sec2}, we define 3- and 4-real symbol decodable unitary weight STBCs and explain the notion of normalization and its use in our analysis.  In Section \ref{sec3}, we present the main result of this paper, a tight upper bound on the rates of 3- and 4-real symbol decodable $2^{a}\times2^{a}$ UWDs. In Section \ref{sec4}, signal sets achieving full-diversity have been identified for the STBCs given in Section \ref{sec3}. Concluding remarks and scope for further work constitute Section \ref{sec5}.

\textit{\textbf{Notations:}}
$\mathbb{R}$ and $\mathbb{C}$ denote the field of real and complex numbers respectively. The set of purely imaginary numbers is represented by $img(\mathbb{C})$. $GL(n,\mathbb{C})$ denotes the group of invertible matrices of size $n\times n$ with complex entries. For any complex matrix $A$, $A^{T}$ and $A^{H}$ represent the Transpose and Hermitian of $A$ respectively. $I_{n}$ and $0_{n}$ represent the $n \times n$ identity matrix and the zero matrix, respectively. For a set $S$, $|S|$ denotes the cardinality of $S$. The Frobenius norm is denoted by $\Vert.\Vert_F$.  For sets $A_{1}$ and $A_{2}$, the Cartesian product of $A_{1}$ and $A_{2}$ is denoted by $A_{1}\times A_{2}$. For a complex number $Z$, complex conjugate is $Z^{*}$. Also, $j$ represents $\sqrt{-1}$ unless it is used as a subscript or index of some quantity or as a running variable. Bold face small letters (ex: ${\bf a}$) represent vectors.
\section{Representation of $\lambda-$real symbol decodable unitary weight STBCs} \label{sec2}
 In this section, we give a representation of $\lambda-$real symbol or  $g$-group decodable STBCs. Any $n \times n$ codeword matrix $S$ of a linear STBC $\cal{S}$ with $g$ groups is represented as
 \begin{equation*}
    {S}=\sum_{i=0}^{g-1}\sum_{j=1}^{\lambda}x_{ij}{A}_{i{j}}
  \end{equation*}
     for $\lambda$-real symbol decodable STBCs, where $\lambda=\frac{K}{g}$. We consider $\lambda=3$ and 4. All the $K$ matrices ($A_{ij},~0\leq i \leq g-1,~1\leq j \leq \lambda$) have to be linearly independent over $\mathbb{R}$.

For a $g$-group decodable STBC $\cal{S}$, a set of necessary and sufficient conditions on the weight matrices are (from (\ref{matrix_condn})),
\begin{equation}\label{nasc}
A_{i{j_1}}^{H}A_{k{j_2}}+A_{k{j_2}}^{H}A_{i{j_1}}=0,
\end{equation}
for $0\leq i\not=k \leq g-1$ and $1\leq  j_1,j_2 \leq \lambda$.
UWDs also satisfy the following criteria
\begin{equation}\label{nasc1}
{A}_{ij}^{H}{A}_{ij}=I_{n}~\text{for}~0\leq i \leq g-1 ~\text{and}~ 1\leq j \leq \lambda.
\end{equation}
\begin{lemma}[\cite{KSR}]{\label{normal}}
Let $\cal{S}$ be a unitary-weight STBC (i.e. obtained from a UWD) and consider the STBC ${\cal{S}}_{U}\triangleq \{U{S}|{S} \in {\cal{S}} \}$, where $U$ is any unitary matrix. Then if $\cal{S}$ satisfies conditions (\ref{nasc}) and (\ref{nasc1}), then so does ${\cal{S}}_{U}$. Further, both the codes have the same coding gain for any signal set $\cal{A}$.
\end{lemma}

The STBCs $\cal{S}$ and ${\cal{S}}_{U}$ are said to be \textit{equivalent}. To simplify our analysis of unitary weight STBCs, we make use of normalization as described below. Let $\cal{S}$ be a unitary weight STBC and let its codeword matrix $S$ be expressed as
 \begin{equation*}
    {S}=\sum_{i=0}^{g-1}\sum_{j=1}^{\lambda}x_{ij}{\hat{A}_{i{j}}}.
  \end{equation*}
Consider the code ${\cal{S}}_{N}\triangleq \{{\hat{A}_{01}^{H}}{S}|{S}\in \cal{S} \}$. Clearly, from Lemma 1, ${\cal{S}}_{N}$ is equivalent to $\cal{S}$. The weight matrices of ${\cal{S}}_{N}$ are
\begin{equation*}
A_{ij}={\hat{A}_{01}^{H}}A_{ij}'~\text{for}~0\leq i \leq g-1 ~\text{and}~ 1\leq j \leq \lambda.
\end{equation*}

We call the code ${\cal{S}}_{N}$ to be the normalized code of $\cal{S}$. In general, any unitary-weight STBC with one of its weight matrices being the identity matrix is called \textit{normalized} unitary-weight STBC. Studying unitary-weight STBCs becomes simpler by studying the normalized unitary-weight STBCs. Now, the conditions presented in (\ref{nasc}) and (\ref{nasc1}) can be written as
\begin{align}\label{nascm}
A_{ij}^{H}&=-A_{ij} ~~ (\text{equivalently}~ A_{ij}^{2}=-I_{n})\\
\label{nascm1} A_{0{j_1}}^{H}A_{i{j_2}}&=A_{i{j_2}}A_{0{j_1}},
\mbox{ for $i\not=0$ and $1\leq j, j_1,j_2 \leq \lambda$}
\end{align}  and
\begin{align}\label{nascm3}
A_{i{j_1}}A_{k{j_2}}=-A_{k{j_2}}A_{i{j_1}},
\end{align}
for $1\leq i\not=k \leq g-1$ and $1\leq  j_1,j_2 \leq \lambda$.

The grouping of weight matrices is shown below.
\begin{center}
 \begin{tabular}{c|c|c|c|c}
 $A_{01}=I_{n}$ & $A_{11}$ & $A_{21}$ & $\ldots $&$ A_{{(g-1)}1}$\\
 \hline
 $A_{02} $& $A_{12}$ & $A_{22}$ & $\ldots $&$ A_{{(g-1)}2}$\\
 \hline
 $\vdots $& $\vdots$ & $\vdots$ & $\ddots $&$ \vdots$\\
 \hline
 $A_{0{\lambda}} $& $A_{1{\lambda}}$ & $A_{2{\lambda}}$ & $\ldots $&$ A_{({g-1}){\lambda}}$
 \end{tabular}
 \end{center}
\section{An upper bound on the rate of 3-and 4-real symbol decodable unitary weight STBCs}\label{sec3}
In this section, we determine the upper bound on the rate of 3-and 4-real symbol decodable $2^{a}\times2^{a}$ UWDs and also give a construction scheme to obtain designs meeting this upper bound. To do so, we make use of the following lemmas regarding matrices of size $n\times n$.
\begin{lemma}[\cite{ShM}]{\label{SM}}
Consider $n\times n$ matrices with complex entries.
\begin{enumerate}
  \item If $n=2^{a}n_{0}$, with $n_{0}$ odd, then there are $l$ elements of $GL(n,\mathbb{C})$ that anti-commute pairwise if and only if $l\leq 2a+1$.
  \item If $n=2^{a}$ and matrices $F_{1},\ldots,F_{2a}$ anti-commute pairwise, then the set of products $F_{i_1}F_{i_2}\cdots F_{i_s}$ with $1\leq i_1 < \cdots <i_s \leq 2a$  along with $I_{n}$ forms a basis for the $2^{2a}$ dimensional space of all $n\times n$ matrices over $\mathbb{C}$. In each case $F_{i}^{2}$ is a scalar matrix (i.e. $F_{i}^{2}=cI_{n},~\text{ where } c\in {\mathbb{C}}$).
\end{enumerate}
\end{lemma}

Let $F_{1},\ldots,F_{2a}$ be anti-commuting, anti-Hermitian, unitary matrices (so that $F_{i}^{2}=-I_{n},~ i=1,2,\cdots,2a$). We can get these matrices from matrix realizations of Clifford algebras and is given in \cite{TiH}.

\begin{lemma}[\cite{KSR}]{\label{Shiparo}}
The product $F_{i_1}F_{i_2}\cdots F_{i_s}$ with $1\leq i_1 < \cdots <i_s\leq 2a$ squares to $(-1)^{\frac{s(s+1)}{2}}I_{n}$.
\end{lemma}

\begin{lemma}[{\cite{KSR}}]
Let $\Omega_{1}=\{F_{i_1},F_{i_2},\cdots,F_{i_s}\}$ and $\Omega_{2}=\{F_{j_1},F_{j_2},\cdots,F_{j_r}\}$ with $1\leq i_1 < \cdots <i_s\leq 2a$ and $1\leq j_1 < \cdots <j_r\leq 2a$. Let $|\Omega_{1}\cap\Omega_{2}|=p$. Then, the product matrix $F_{i_1}F_{i_2}\cdots F_{i_s}$ commutes with $F_{j_1}F_{j_2}\cdots F_{j_r}$, if exactly one of the following is satisfied, and anti-commutes otherwise.
\begin{enumerate}
  \item $r$, $s$ and $p$ are all odd.
  \item The product $rs$ is even and $p$ is even (including 0).
\end{enumerate}
\end{lemma}

\begin{lemma}[\cite{KSR}]\label{thmssd}
The maximum rate in cspcu of a $2^{a} \times 2^{a}$ unitary-weight SSD code is $\frac{a}{2^{a-1}}$.
\end{lemma}

The above Lemma is equivalent to showing that $2a$ is the maximum number of groups possible for 2-real symbol (1-complex symbol) decodable $2^{a} \times 2^{a}$ UWD.

Though the following theorem (proof given in Appendix A) does not give a tight bound, Theorem \ref{thm3sym1} and Theorem \ref{thmcode} are used in Theorem  \ref{thm3sym2}, to get a tight bound.
\begin{theorem}\label{thm3sym1}
For a 3-real symbol decodable $2^{a} \times 2^{a}$ UWD, the rate in cspcu is upper bounded by $\frac{3(2a-1)}{2^{a+1}}$, which is not tight.
\end{theorem}

\begin{theorem}\label{thmcode}
There exists a rate $\frac{a-1}{2^{a-2}}$ cspcu, 4-real symbol decodable $2^{a} \times 2^{a}$ UWD for $a\ge2$.
\end{theorem}

\begin{proof}
     Proof is by explicit construction. This construction is based on the proof of Theorem 6 in \cite{RaR}.

    For $ a\ge2$ let $m=2^{a-2}$ and $n=2^a.$ Then from Lemma \ref{SM}, for $m\times m$ matrices, we can have $2(a-2)+1$  anti-Hermitian and anti-commuting unitary matrices. Let them be $E_{1},E_{2},\cdots,E_{2a-3}$. Let $A_{01}=I_{n}$      and for $1\le i\le 2a-3,$
{\scriptsize
\begin{align}
\nonumber A_{i1} =    \begin{bmatrix}
       E_{i} & 0_{m} & 0_{m} & 0_{m} \\
       0_{m} & E_{i} & 0_{m} & 0_{m} \\
       0_{m} & 0_{m} & E_{i} & 0_{m} \\
       0_{m} & 0_{m} & 0_{m} & E_{i} \\
     \end{bmatrix};
\nonumber A_{02} =    \begin{bmatrix}
       I_{m} & 0_{m} & 0_{m} & 0_{m} \\
       0_{m} & -I_{m} & 0_{m} & 0_{m} \\
       0_{m} & 0_{m} & I_{m} & 0_{m} \\
       0_{m} & 0_{m} & 0_{m} & -I_{m} \\
     \end{bmatrix};
\end{align}
\begin{align}
\nonumber A_{03} =\begin{bmatrix}
       I_{m} & 0_{m} & 0_{m} & 0_{m} \\
       0_{m} & I_{m} & 0_{m} & 0_{m} \\
       0_{m} & 0_{m} & -I_{m} & 0_{m} \\
       0_{m} & 0_{m} & 0_{m} & -I_{m} \\
     \end{bmatrix};
\nonumber A_{04} =\begin{bmatrix}
       I_{m} & 0_{m} & 0_{m} & 0_{m} \\
       0_{m} & -I_{m} & 0_{m} & 0_{m} \\
       0_{m} & 0_{m} & -I_{m} & 0_{m} \\
       0_{m} & 0_{m} & 0_{m} & I_{m} \\
     \end{bmatrix};
\end{align}}
 $A_{i2}=A_{02}A_{i1}$;~ $A_{i3}=A_{03}A_{i1}$ and $A_{i4}=A_{04}A_{i1}$.

It can be easily seen that these matrices satisfy the conditions (\ref{nasc}), (\ref{nascm}) and (\ref{nascm3}). So, we have constructed a 4-real symbol $2a-2$ group decodable UWD. The rate of this design is $\frac{4(2a-2)}{2^{a+1}}=\frac{a-1}{2^{a-2}}$ cspcu.
\end{proof}
\begin{corollary}\label{corcode}
There exists a rate $\frac{3(a-1)}{2^{a}}$ cspcu, 3-real symbol decodable $2^{a} \times 2^{a}$ UWD for $a\ge2$.
\end{corollary}
\begin{proof}
Straightforward from Theorem \ref{thmcode}, by removing $A_{i_4}$ matrices.
\end{proof}

\begin{theorem}\label{thm3sym2}
For a 3-real symbol decodable $2^{a} \times 2^{a}$ UWD, the rate in cspcu is tightly upper bounded by $\frac{3(a-1)}{2^{a}}$.
\end{theorem}

\begin{proof}
Let $g$ be the number of groups.
From Theorem \ref{thm3sym1} and Corollary \ref{corcode}, it is enough to show that $g=2a-1$ is not possible. To prove this, consider the following grouping of weight matrices:
\begin{center}
 \begin{tabular}{c|c|c|c|c}
 $I_{n}$ & $F_{1}$ &$ F_{2}$ & \ldots & $F_{2a-2}$\\
 \hline
 $A_{02} $& $A_{12}$ & $A_{22}$ & $\ldots $&$ A_{{(2a-2)}2}$\\
 \hline
 $A_{03} $& $A_{13}$ & $A_{23}$ & $\ldots $&$ A_{{(2a-2)}3}$
 \end{tabular}
 \end{center}

The theorem is proved in the following 5 steps, the proof for all of which is given in Appendix B:
\begin{description}
  \item[Step 1:] ~Finding a relation between coefficients of $A_{1i}$ and $A_{2j}$ ($i,j\in\{2,3\}$) when expanded in terms of the basis of Lemma \ref{SM}.
  \item[Step 2:] ~Finding a relation between coefficients of $A_{0i}$, $A_{1j}$ and $A_{2k}$ ($i,j,k\in\{2,3\}$) when expanded in terms of the basis of Lemma \ref{SM}.
  \item[Step 3:] ~Showing that there is a possibility of 7 types of solutions that takes into account the relations in Step 1.
  \item[Step 4:] ~Showing that after including the relations  from Step 2 also,  there is a possibility of 7 types of solutions.
  \item[Step 5:] ~None of the 7 solutions in Step 4 is possible.
\end{description}

\end{proof}

\begin{theorem}\label{thm4}
For a 4-real symbol decodable $2^{a} \times 2^{a}$ UWD, the rate in cspcu is tightly upper bounded by $\frac{(a-1)}{2^{a-2}}$.
\end{theorem}
\noindent
The proof is given in Appendix C.
\begin{example}\label{example1}
Consider a $4\times 4$ UWD with weight matrices $A_{01}=I_{4}$, $A_{11}=jI_{4}$ and

{\footnotesize
\begin{align*}
 A_{02} = \left[ \begin{array}{rrrr}
  1 & 0 & 0 & 0 \\
  0 & -1 & 0 & 0 \\
  0 & 0 & 1 & 0 \\
  0 & 0 & 0 & -1 \\
\end{array} \right]
 A_{03} =  \left[ \begin{array}{rrrr}
  1 & 0 & 0 & 0 \\
  0 & 1 & 0 & 0 \\
  0 & 0 & -1 & 0 \\
  0 & 0 & 0 & -1 \\
\end{array} \right]
\end{align*}
\begin{align*}
 A_{04} =   \left[ \begin{array}{rrrr}
  1 & 0 & 0 & 0 \\
  0 & -1 & 0 & 0 \\
  0 & 0 & -1 & 0 \\
  0 & 0 & 0 & 1 \\
\end{array} \right],
\end{align*}}
\noindent
with $A_{12}=jA_{02}$, $A_{13}=jA_{03}$ and $A_{14}=jA_{04}$. The codeword matrix $S({\bf x}_{0},{\bf x}_{1})$ is given by (${\bf x}_{i}=[x_0~ x_1~ \cdots x_{\lambda}]$)
\begin{align*}
 S({\bf x}_{0},{\bf x}_{1})=\sum_{i=0}^{1}\sum_{j=1}^{4}x_{ij}{A}_{ij} = \left[ \begin{array}{rrrr}
  Z_{1} & 0 & 0 & 0 \\
  0 & Z_{2} & 0 & 0 \\
  0 & 0 & Z_{3} & 0 \\
  0 & 0 & 0 & Z_{4} \\
\end{array} \right],
\end{align*}
for 4-real symbol decodable UWD and $S({\bf x}_{0},{\bf x}_{1})$ is given by
\begin{align*}
 S({\bf x}_{0},{\bf x}_{1})=\sum_{i=0}^{1}\sum_{j=1}^{3}x_{ij}{A}_{ij} = \left[ \begin{array}{rrrr}
  Z_{5}  & 0 & 0 & 0 \\
  0 & Z_{6} & 0 & 0 \\
  0 & 0 & Z_{7} & 0 \\
  0 & 0 & 0 &  Z_{8}\\
\end{array} \right],
\end{align*}
for 3-real symbol decodable UWD, where,
\begin{align*}
 Z_{1}&=x_{01} +x_{02}+x_{03} +x_{04}+jx_{11} +jx_{12}+jx_{13} +jx_{14}\\
 Z_{2}&=x_{01} -x_{02}+x_{03} -x_{04}+jx_{11} -jx_{12}+jx_{13} -jx_{14}\\
 Z_{3}&=x_{01} +x_{02}-x_{03} -x_{04}+jx_{11} +jx_{12}-jx_{13} -jx_{14}\\
 Z_{4}&=x_{01} -x_{02}-x_{03} +x_{04}+jx_{11} -jx_{12}-jx_{13} +jx_{14}\\
 Z_{5}&=x_{01} +x_{02}+x_{03} +jx_{11} +jx_{12}+jx_{13}\\
 Z_{6}&=x_{01} -x_{02}+x_{03} +jx_{11} -jx_{12}+jx_{13}\\
 Z_{7}&=x_{01} +x_{02}-x_{03}+jx_{11} +jx_{12}-jx_{13}\\
 Z_{8}&=x_{01} -x_{02}-x_{03} +jx_{11} -jx_{12}-jx_{13}.
\end{align*}
 It is easily checked that the weight matrices above  satisfy (\ref{nascm}) to (\ref{nascm3}). So, for this 4-real symbol decodable UWD, rate is 1 cspcu and for 3-real symbol decodable UWD, rate is $\frac{3}{4}$ cspcu.
\end{example}
\begin{example}\label{example2}
Consider a $8\times 8$ UWD with codeword matrix $S({\bf x}_{0},{\bf x}_{1},{\bf x}_{2},{\bf x}_{3})=\sum_{i=0}^{3}\sum_{j=1}^{4}x_{ij}{A}_{ij}$ given by

{\small\begin{align*}
  \left[
  \begin{array}{rrrrrrrr}
    Z_{1} & Z_{2} & Z_{3} & Z_{4} & Z_{5} & Z_{6} & Z_{7} & Z_{8} \\
    Z_{4} & -Z_{3} & -Z_{2} & Z_{1} & Z_{8} & -Z_{7} & -Z_{6} & Z_{5} \\
    Z_{5} & Z_{6} & Z_{7} & Z_{8} & Z_{1} & Z_{2} & Z_{3} & Z_{4} \\
    Z_{8} & -Z_{7} & -Z_{6} & Z_{5} & Z_{4} & -Z_{3} & -Z_{2} & Z_{1} \\
    Z_{2}^{*} & -Z_{1}^{*} & Z_{4}^{*} & -Z_{3}^{*} & Z_{6}^{*} & -Z_{5}^{*} & Z_{8}^{*} & -Z_{7}^{*} \\
    Z_{3}^{*} & Z_{4}^{*} & -Z_{1}^{*} & -Z_{2}^{*} & Z_{7}^{*} & Z_{8}^{*} & -Z_{5}^{*} & -Z_{6}^{*} \\
    Z_{6}^{*} & -Z_{5}^{*} & Z_{8}^{*} & -Z_{7}^{*} & Z_{2}^{*} & -Z_{1}^{*} & Z_{4}^{*} & -Z_{3}^{*} \\
    Z_{7}^{*} & Z_{8}^{*} & -Z_{5}^{*} & -Z_{6}^{*} & Z_{3}^{*} & Z_{4}^{*} & -Z_{1}^{*} & -Z_{2}^{*} \\
  \end{array}
\right].
\end{align*}}
where, $Z_{1}=x_{01}+jx_{11}$, $Z_{2}=x_{21}+jx_{31}$, $Z_{3}=x_{22}+jx_{32}$, $Z_{4}=x_{02}+jx_{12}$, $Z_{5}=x_{03}+jx_{13}$, $Z_{6}=x_{23}+jx_{33}$, $Z_{7}=x_{24}+jx_{34}$ and $Z_{8}=x_{04}+jx_{14}$. By calculating $S^{H}S$, we can easily say that, this design is 4-real symbol decodable UWD. Rate of this design is 1 cspcu. By assigning $x_{i4}=0~(i\in\{0,1,2,3\})$, we get 3-real symbol decodable UWD with rate $\frac{3}{4}$ cspcu.
\end{example}
\section{Diversity and Coding Gain}\label{sec4}
In this Section, we show that for the code shown in Theorem \ref{thmcode}, full diversity is achievable for 4-real symbol decodable STBCs with $n=2^{a}$ antennas.  Also, expressions for coding gain are presented.

Let, ${\bf x}=[x_1~x_{2}~x_{3}~x_{4}]^{T}$ take values from a finite signal set ${\cal{B}.}$ The differential signal set $\bigtriangleup\cal{B}$ of signal set $\cal{B}$ is defined as
\begin{equation*}
{\bigtriangleup\cal{B}}=\{\bigtriangleup {\bf xx}'= {\bf x}-{\bf x}'|{\bf x,x'}\in {\cal{B}}\}.
\end{equation*}
For a differential signal set $\bigtriangleup\cal{B}$, $\forall \bigtriangleup {\bf xx}' \in \bigtriangleup\cal{B}$ Let,
{\small
\begin{align*}
\Psi(\bigtriangleup {\bf xx'} ) = &\left[((\bigtriangleup {\bf xx'})_1+(\bigtriangleup {\bf xx'})_{2}+(\bigtriangleup {\bf xx'})_{3}+(\bigtriangleup {\bf xx'})_{4})^{2}\right]^{\frac{1}{4}}\\
&\left[((\bigtriangleup {\bf xx'})_{1}-(\bigtriangleup {\bf xx'})_{2}+(\bigtriangleup {\bf xx'})_{3}-(\bigtriangleup {\bf xx'})_{4})^{2}\right]^{\frac{1}{4}}\\
 &\left[((\bigtriangleup {\bf xx'})_{1}+(\bigtriangleup {\bf xx'})_{2}-(\bigtriangleup {\bf xx'})_{3}-(\bigtriangleup {\bf xx'})_{4})^{2}\right]^{\frac{1}{4}}\\
 &\left[((\bigtriangleup {\bf xx'})_{1}-(\bigtriangleup {\bf xx'})_{2}-(\bigtriangleup {\bf xx'})_{3}+(\bigtriangleup {\bf xx'})_{4})^{2}\right]^{\frac{1}{4}}.
\end{align*}
}
\begin{theorem}\label{div_thm}
If $\cal{B}$ is the signal set, from which the variables of a group take values from, for the STBC of Theorem \ref{thmcode}, then, the code achieves full-diversity if the differential signal set $\bigtriangleup\cal{B}$ satisfies
\begin{equation}\label{signal_condn}
\Psi(\bigtriangleup {\bf xx'}) >0, ~~~~~~~~\forall~~~ (\bigtriangleup {\bf xx'}\not=0) \in \bigtriangleup\cal{B}.
\end{equation}

\end{theorem}
\begin{proof}
Proof available in Appendix D.
\end{proof}

Now, we will find the signal set $\cal{B}$, which gives full-diversity for the STBC of Theorem \ref{thmcode}.

Let us define a new vector variable $\bigtriangleup q\triangleq [\bigtriangleup q_{1}~\bigtriangleup q_{2}~\bigtriangleup q_{3}~\bigtriangleup q_{4}]^{T}$ as $\bigtriangleup q = P \bigtriangleup {\bf xx'},$
where
\begin{align*}
 P= \left[ \begin{array}{rrrr}
  1 & 1 & 1 & 1 \\
  1 & -1 & 1 & -1 \\
  1 & 1 & -1 & -1 \\
  1 & -1 & -1 & 1 \\
\end{array} \right].
\end{align*}
 Now, $\bigtriangleup_{min}$ (from  Appendix D) can be re-written as (since $P$ is invertible $\bigtriangleup q=0$ iff $\bigtriangleup {\bf xx'}=0$)
\begin{equation*}
\bigtriangleup_{min}= \min_{\bigtriangleup q \not= 0} [(\bigtriangleup q_{1})^{2}(\bigtriangleup q_{2})^{2}(\bigtriangleup q_{3})^{2}(\bigtriangleup q_{4})^{2}]^{\frac{n}{4}}.
\end{equation*}

 To achieve full diversity we need $\bigtriangleup_{min}>0$. This can be achieved if $\bigtriangleup q_{i}\not=0,~\forall~1 \leq i \leq 4$. And, this can be guaranteed by letting ${\bf x}=[x_{1}~x_{2}~x_{3}~x_{4}]^{T}$ take values from $P^{-1}{\cal{G}}_{4}\mathbb{Z}^{4},$ where, ${\cal{G}}_{4}$ is the generator matrix of a 4-dimensional lattice designed to maximize the product distance \cite{EFE},\cite{FDRV}, and the coding gain $\delta_{min}$ is given by

{\small\begin{equation}
\label{codingain}\delta_{min}=(\bigtriangleup_{min})^{\frac{1}{n}}= \min_{\bigtriangleup q \not= 0} [(\bigtriangleup q_{1})^{2}(\bigtriangleup q_{2})^{2}(\bigtriangleup q_{3})^{2}(\bigtriangleup q_{4})^{2}]^{\frac{1}{4}}.
\end{equation}}
\noindent
The right hand side of (\ref{codingain}) can be obtained from \cite{EFE}, \cite{FDRV}.

\subsection{Calculation of Diversity and Coding gain with examples}
Let, the signal set $\cal{B}$ be obtained for $[\pm1~\pm1~\pm1~\pm1]^{T}\in\mathbb{Z}^{4}$ from $P^{-1}{\cal{G}}_{4}\mathbb{Z}^{4}$. Here, $P^{-1}=\frac{1}{4}P$ and from \cite{FDRV}

{\small\begin{align*}
{\cal{G}}_{4}=\left[
  \begin{array}{rrrr}
    -0.3664 & -0.7677 & 0.4231 & 0.3121 \\
    -0.2264 & -0.4745 & -0.6846 & -0.5050 \\
    -0.4745 & 0.2264 & -0.5050 & 0.6846 \\
    -0.7677 & 0.3664 & 0.3121 & -0.4231 \\
  \end{array}
\right].
\end{align*}}

Let ${\bf x} \in \cal{B}$, be written as ${\bf x}=eP^{-1}{\cal{G}}_{4}{\bf z}$, for some ${\bf z}\in [\pm1~\pm1~\pm1~\pm1]^{T}$, where, $e$ is used for normalizing the average energy. Define $E({\bf x})=||{\bf x}||^{2}=x_{1}^{2}+x_{2}^{2}+x_{3}^{2}+x_{4}^{2}$. Then
\begin{align}
\nonumber E({\bf x})=&\frac{1}{4}e^{2}{\bf z}^{T}{\cal{G}}_{4}^{T}PP^{-1}{\cal{G}}_{4}{\bf z}=\frac{1}{4}e^{2}{\bf z}^{T}{\bf z}
=e^{2}.
\end{align}
Here, we used the fact that ${\cal{G}}_{4}^{T}{\cal{G}}_{4}=I_{4}$ from \cite{FDRV} and ${\bf z}^{T}{\bf z}=4$. So, $E({\bf x})=e^{2}=E$ for all ${\bf x}$ such that ${\bf x}\in \cal{B}$.

Let, $\bigtriangleup u={\cal{G}}_{4}\bigtriangleup v$, where $\bigtriangleup u=[\bigtriangleup u_{1}~\bigtriangleup u_{2}~\bigtriangleup u_{3}~\bigtriangleup u_{4}]$ and $\bigtriangleup v=[\bigtriangleup v_{1}~\bigtriangleup v_{2}~\bigtriangleup v_{3}~\bigtriangleup v_{4}]$. Let $\bigtriangleup v_{i}\in\{-2,0,2\}$. Define,
\begin{align*}
\delta_{u}\triangleq\min_{\bigtriangleup u \not= 0} [(\bigtriangleup u_{1})^{2}(\bigtriangleup u_{2})^{2}(\bigtriangleup u_{3})^{2}(\bigtriangleup u_{4})^{2}]^{\frac{1}{4}}=0.6503.
\end{align*}

Now, the following example calculates the coding gains for a $4\times4$ and $8\times8$, 4-real symbol decodable Unitary-weight STBCs of Example \ref{example1} and  Example \ref{example2}.
\begin{example}\label{example3}
For the codeword $S$ of Example \ref{example1} (for 4-real symbol decodable UWD), average energy $E_{avg}$ is given by
\begin{align*}
E_{avg}=\frac{1}{16}\mathbb{E}(||S||_{F}^{2})=&\frac{1}{16}\mathbb{E}(|Z_{1}|^{2}+|Z_{2}|^{2}+|Z_{3}|^{2}+|Z_{4}|^{2})\\
=&\frac{1}{4}\mathbb{E}(||{\bf x}_{0}||^{2}+||{\bf x}_{1}||^{2})=\frac{E}{2},
\end{align*}
 where $\mathbb{E}$ is over all possible information symbols. For $E_{avg}=1$ to be satisfied, $E=2$ and $e=\sqrt{2}$.
From (\ref{codingain}), coding gain is given by ${e}^{2}\delta_{u}=1.3006=1.1414dB$. Since $\delta_{min}>0$, this STBC has full diversity.

For the codeword $S$ of Example \ref{example2}, the  average energy $E_{avg}$ is given by
{\small\begin{align*}
E_{avg}=&\frac{1}{64}\mathbb{E}(||S||_{F}^{2})=\frac{1}{8}\mathbb{E}(\sum_{i=1}^{8} |Z_{i}|^{2}) = \frac{1}{8}\mathbb{E}(\sum_{j=0}^{3}||{\bf x}_{j}||^{2})=\frac{1}{2}E,
\end{align*}}
 where $\mathbb{E}$ is over all possible information symbols. For $E_{avg}=1$ to be satisfied, $E={2}$ and $e=\sqrt{2}$.

Let, for a complex number $Z_{i}=\Re\{Z_{i}\}+j\Im\{Z_{i}\}$, $\bigtriangleup Z_{i} =\Re\{ \bigtriangleup Z_{i}\}+j\Im\{\bigtriangleup Z_{i}\}$, where $\Re\{Z_{i}\}$, $\Im\{Z_{i}\}$ are the real and imaginary parts of $Z_{i}$ respectively. To find coding gain we need to find $det[(\bigtriangleup {S})^{H}(\bigtriangleup {S})]=det[(\bigtriangleup {S})(\bigtriangleup {S})^{H}]$.
\begin{align*}
(\bigtriangleup {S})(\bigtriangleup {S})^{H}= \left [\begin{array}{cc}
                                                \Phi(t,\alpha,\beta,\gamma) & 0_{4} \\
                                                0_{4} & \Phi(t,-\alpha,\beta,-\gamma)                                              \end{array} \right ].
\end{align*}
Here,
\begin{align*}
 \Phi(t,\alpha,\beta,\gamma)= \left[ \begin{array}{rrrr}
  t & \alpha & \beta & \gamma \\
  \alpha & t & \gamma & \beta \\
  \beta & \gamma & t & \alpha \\
  \gamma & \beta & \alpha & t \\
\end{array} \right],
\end{align*}
where,
{\small\begin{align*}
t&=\sum_{i=1}^{8} |\bigtriangleup Z_{i}|^{2}, \\
\alpha&=2\Re\{\bigtriangleup Z_{1} \bigtriangleup Z_{4}^{*} - \bigtriangleup Z_{2}\bigtriangleup Z_{3}^{*}+ \bigtriangleup Z_{5}\bigtriangleup Z_{8}^{*} - \bigtriangleup Z_{6}\bigtriangleup Z_{7}^{*}\},\\
\beta&=2\Re\{\bigtriangleup Z_{1}\bigtriangleup Z_{5}^{*} + \bigtriangleup Z_{2}\bigtriangleup Z_{6}^{*} + \bigtriangleup Z_{3}\bigtriangleup Z_{7}^{*}+\bigtriangleup Z_{4}\bigtriangleup Z_{8}^{*}\},\\
\gamma&=2\Re\{\bigtriangleup Z_{1}\bigtriangleup Z_{8}^{*} - \bigtriangleup Z_{2}\bigtriangleup Z_{7}^{*}-\bigtriangleup Z_{3}\bigtriangleup Z_{6}^{*}+\bigtriangleup Z_{4}\bigtriangleup Z_{5}^{*}\}.
\end{align*}}
\begin{align*}
det[\Phi(t,\alpha,\beta,\gamma)]=&(t+\gamma+\alpha+\beta)(t+\gamma-\alpha-\beta)\\
&(t-\gamma+\alpha-\beta)(t-\gamma-\alpha+\beta)\\
=&det[\Phi(t,-\alpha,\beta,-\gamma)]
\end{align*}
Now, 
\begin{align*}
det[(\bigtriangleup {S})^{H}(\bigtriangleup {S})]=& det[\Phi(t,\alpha,\beta,\gamma)]\times det[\Phi(t,-\alpha,\beta,-\gamma)]\\
 =&(det[\Phi(t,\alpha,\beta,\gamma)])^{2}.
\end{align*}
As in (\ref{divdet}), here too $det[(\bigtriangleup {S})^{H}(\bigtriangleup {S})]$ is a product of sum of squares of real numbers, so,
\begin{align*}
 &det[(\bigtriangleup {S})^{H}(\bigtriangleup {S})]\geq \\&
 \left[((\bigtriangleup {\bf x}_{i}{\bf x}_{i}')_{1}+(\bigtriangleup {\bf x}_{i}{\bf x}_{i}')_{2}+(\bigtriangleup {\bf x}_{i}{\bf x}_{i}')_{3}+(\bigtriangleup {\bf x}_{i}{\bf x}_{i}')_{4})\right]^{4}\\
 &\left[((\bigtriangleup {\bf x}_{i}{\bf x}_{i}')_{1}-(\bigtriangleup {\bf x}_{i}{\bf x}_{i}')_{2}+(\bigtriangleup {\bf x}_{i}{\bf x}_{i}')_{3}-(\bigtriangleup {\bf x}_{i}{\bf x}_{i}')_{4})\right]^{4}\\
&\left[((\bigtriangleup {\bf x}_{i}{\bf x}_{i}')_{1}+(\bigtriangleup {\bf x}_{i}{\bf x}_{i}')_{2}-(\bigtriangleup {\bf x}_{i}{\bf x}_{i}')_{3}-(\bigtriangleup {\bf x}_{i}{\bf x}_{i}')_{4})\right]^{4}\\
&\left[((\bigtriangleup {\bf x}_{i}{\bf x}_{i}')_{1}-(\bigtriangleup {\bf x}_{i}{\bf x}_{i}')_{2}-(\bigtriangleup {\bf x}_{i}{\bf x}_{i}')_{3}+(\bigtriangleup {\bf x}_{i}{\bf x}_{i}')_{4}) \right]^{4},
\end{align*}
for some $0\leq i \leq 3$. And, $\bigtriangleup_{min}$ occurs when all but one among $\bigtriangleup {\bf x}_{i}{\bf x}_{i}',~0\leq i \leq 3$ are zeros. So, from (\ref{codingain}), coding gain is given by ${e}^{2}\delta_{u}=1.3006=1.1414dB$.
 Since $\delta_{min}>0$, this STBC has full diversity.
\end{example}

\section{Discussion and Concluding remarks}\label{sec5}
In this paper, we have shown that, the maximum rate achieved by 3- and 4-real symbol  ML decodable $2^{{a}}\times2^{{a}}$ UWDs is $\frac{3(a-1)}{2^{a}}$ and $\frac{4(a-1)}{2^{a}}$ cspcu respectively. We have also given a \textit{STBC} which achieves the maximum rate. And, also shown that this \textit{STBC} can achieve full-diversity for rotated lattice constellations. Possible directions for further research are:
\begin{enumerate}
  \item A general upper bound on the rate of the $\lambda$ real symbol decodable UWDs is yet to be found.
  \item Even though maximum rate possible for $g$-group decodable CUWDs is found, it is not found for general UWDs.
\end{enumerate}
These could possibly be the future direction of research.
\section*{acknowledgement}
The authors wish to thank K. Pavan Srinath for his useful reviews on proofs of the theorems. We also wish to thank G. Abhinav for proof reading the paper. This work was supported partly by the DRDO-IISc program on Advanced Research in Mathematical Engineering through a research grant and partly by the INAE Chair Professorship grant to B. S. Rajan.

\newpage

\section*{Appendix A}
\begin{center}
{\bf Proof of Theorem \ref{thm3sym1}}
\end{center}
Let $g$ be the number of groups. The rate is given by $\frac{3g}{2^{a+1}}$ cspcu. It is enough to show that $g$ is upper bounded by $2a-1$. If $g$ is more than $2a$ we can remove one weight matrix from each group and get 2-real symbol  decodable UWD, which contradicts Lemma \ref{thmssd}.  So, $g \leq 2a$. To prove that $g$ cannot be $2a$ let us consider the following weight matrices,
\begin{center}
 \begin{tabular}{c|c|c|c|c}
 $I_{n}$ & $F_{1}$ & $F_{2}$ & \ldots & $F_{2a-1}$\\
 \hline
 $A_{02} $& $A_{12}$ & $A_{22}$ & $\ldots $&$ A_{{(2a-1)}2}$\\
 \hline
 $A_{03} $& $A_{13}$ & $A_{23}$ & $\ldots $&$ A_{{(2a-1)}3}$
 \end{tabular}
 \end{center}
where $A_{01}=I_{n}$ and $A_{j1}=F_{j}$ ($ j=1,2,\cdots,2a-1$).
From Case-3 in the proof of Theorem 1 of \cite{KSR}, $A_{j2}$ can be written as $A_{j2} = \pm m \prod_{i=1,i\not=j}^{2a-1} F_{i}, $ for $ j=1,2,\cdots,2a-1$. Hence, $A_{j3}$ must be equal to $a_{j,1}F_{j}+a_{j,2}F_{2a}+a_{j,4}F_{1}F_{2}\cdots F_{2a}$. But these $A_{j3}$s violate linear independence of weight matrices, so we cannot have $A_{j3}$s for $g=2a$. This completes the proof.

\section*{Appendix B}
\begin{center}
{\bf Proof of Theorem \ref{thm3sym2}}
\end{center}

\textbf{Step 1}:

Let $A_{ik}=\sum_{j=1}^{2^{2a}}a_{ikj}F_{1}^{\lambda_{1,j}}F_{2}^{\lambda_{2,j}}\cdots F_{2a}^{\lambda_{2a,j}},~  \lambda_{m,j}\in \{0,1\},~ m=1,2,\cdots,2a,~ i=0,1,2,\cdots,2a-2,~k=2,3$ and $a_{ikj}\in \mathbb{C}$. This is possible because of Lemma \ref{SM}. Considering $A_{1k}$, since $A_{1k}$ anti-commutes with $F_{2}$, $F_{3}$, $ \cdots$, $F_{2a-2}$, every individual term of $A_{1k}$ must anti-commute with $F_{2}$, $ F_{3}$, $\cdots$, $ F_{2a-2}$. The only matrices from the set $\{F_{1}^{\lambda_{1}}F_{2}^{\lambda_{2}}\cdots F_{2a}^{\lambda_{2a}},~ \lambda_{i}\in \{0,1\} ,~ i=1,2,\cdots,2a\}$ that anti-commute with $F_{2}, F_{3}, \cdots, F_{2a-2}$ are $F_{1}$, $ F_{2a-1}$, $ F_{2a}$, $F_{1}F_{2a-1}F_{2a}$, $F_{1}F_{2}\cdots F_{2a-2}$, $F_{2}F_{3}\cdots F_{2a-1}$, $F_{2}F_{3}\cdots F_{2a-2}F_{2a}$, and $F_{1}F_{2}\cdots F_{2a}$. Hence, let (for $k \in \{2,3\}$)
{
\begin{align*}
A_{1k} = & a_{k1}F_{1}+a_{k2}F_{2a-1}+a_{k3}F_{2a}+a_{k4}F_{1}F_{2a-1}F_{2a}\\
&+a_{k5}F_{1}F_{2}\cdots F_{2a-2}+a_{k6}F_{2}F_{3}\cdots F_{2a-1}\\
&+a_{k7}F_{2}F_{3}\cdots F_{2a-2}F_{2a}+a_{k8}F_{1}F_{2}\cdots F_{2a}.
\end{align*}}
 Similarly, let (for $k \in \{2,3\}$)
{ \begin{align}
\nonumber A_{2k} = & c_{k1}F_{2}+c_{k2}F_{2a-1}+c_{k3}F_{2a}+c_{k4}F_{1}F_{2a-1}F_{2a}\\
\nonumber &+c_{k5}F_{1}F_{2}\cdots F_{2a-2}+c_{k6}F_{1}F_{3}\cdots F_{2a-1}\\
\label{b_rep} &+c_{k7}F_{1}F_{3}\cdots F_{2a-2}F_{2a}+c_{k8}F_{1}F_{2}\cdots F_{2a}.
 \end{align}}

Let, $(F_{1}F_{2}\cdots F_{2a})^{2}=pI$, where $p=-1$ when, $a$ is odd and $p=1$ when, $a$ is even. Further, from (\ref{nascm}), $A_{1k}^{2}=A_{2k}^{2}=-I_{n}$ and from the above equations we have
{\small
\begin{align*}
A_{1k}^2=&-(a_{k1}^{2}+a_{k2}^{2}+a_{k3}^{2}-a_{k4}^{2}+p(a_{k5}^{2}+a_{k6}^{2}+a_{k7}^{2}-a_{k8}^{2}))I_{n}\\
&-2(a_{k1}a_{k4}+pa_{k5}a_{k8})F_{2a-1}F_{2a}\\
&+2(a_{k1}a_{k6}+a_{k2}a_{k5})F_{1}\cdots F_{2a-1}\\
&+2(a_{k1}a_{k7}+a_{k3}a_{k5})F_{1}\cdots F_{2a-2}F_{2a}\\
&+ 2(a_{k2}a_{k4}-pa_{k6}a_{k8})F_{1}F_{2a}\\
&-2(a_{k2}a_{k7}-a_{k3}a_{k6})F_{2}F_{3}\cdots F_{2a}\\
&-2(a_{k3}a_{k4}-pa_{k7}a_{k8})F_{1}F_{2a-1} ,\\
A_{2k}^{2}=&-(c_{k1}^{2}+c_{k2}^{2}+c_{k3}^{2}-c_{k4}^{2}+p(c_{k5}^{2}+c_{k6}^{2}+c_{k7}^{2}-c_{k8}^{2}))I_{n}\\
&-2(c_{k1}c_{k4}+pc_{k5}c_{k8})F_{2a-1}F_{2a}\\
&-2(c_{k1}c_{k6}-c_{k2}c_{k5})F_{1}\cdots F_{2a-1}\\
&-2(c_{k1}c_{k7}-c_{k3}c_{k5})F_{1}\cdots F_{2a-2}F_{2a}\\
&+ 2(c_{k2}c_{k4}+pc_{k6}c_{k8})F_{2}F_{2a}\\
&-2(c_{k2}c_{k7}-c_{k3}c_{k6})F_{1}F_{3}\cdots F_{2a-1}F_{2a}\\
&-2(c_{k3}c_{k4}+pc_{k7}c_{k8})F_{2}F_{2a-1}.
\end{align*}}

Since $I_{n},$ $F_{2a-1}F_{2a},$ $F_{1}F_{2}\cdots F_{2a-1}$, $F_{1}F_{2}\cdots F_{2a-2}F_{2a},$ $F_{1}F_{2a}$, $F_{2}F_{3}\cdots F_{2a},$ $F_{1}F_{2a-1},$ $F_{2}F_{2a},$ $F_{1}F_{3}\cdots F_{2a}$ and $F_{2}F_{2a-1}$ are linearly independent over $\mathbb{C}$, the following equations have to be satisfied.

{\small
\begin{equation}\label{aci1}
\left.
\begin{aligned}
&a_{k1}^{2}+a_{k2}^{2}+a_{k3}^{2}-a_{k4}^{2}+p(a_{k5}^{2}+a_{k6}^{2}+a_{k7}^{2}-a_{k8}^{2})=1;\\
&c_{k1}^{2}+c_{k2}^{2}+c_{k3}^{2}-c_{k4}^{2}+p(c_{k5}^{2}+c_{k6}^{2}+c_{k7}^{2}-c_{k8}^{2})=1;\\
\end{aligned}
\right\}
\end{equation}
}

\begin{equation}\label{aci2}
\left.
\begin{aligned}
&~~~~~~a_{k1}a_{k4}+pa_{k5}a_{k8}=0;~a_{k1}a_{k6}+a_{k2}a_{k5}=0;\\
&~~~~~~a_{k1}a_{k7}+a_{k3}a_{k5}=0;~a_{k2}a_{k4}-pa_{k6}a_{k8}=0;\\
&~~~~~~a_{k2}a_{k7}-a_{k3}a_{k6}=0;~a_{k3}a_{k4}-pa_{k7}a_{k8}=0;\\
&~~~~~~c_{k1}c_{k4}+pc_{k5}c_{k8}=0;~c_{k1}c_{k6}-c_{k2}c_{k5}=0;\\
&~~~~~~c_{k1}c_{k7}-c_{k3}c_{k5}=0;~c_{k2}c_{k4}+pc_{k6}c_{k8}=0;\\
&~~~~~~c_{k2}c_{k7}-c_{k3}c_{k6}=0;~c_{k3}c_{k4}+pc_{k7}c_{k8}=0.
\end{aligned}
\right\}
\end{equation}

Since $A_{ik}^{H}=-A_{ik}, ~i\in\{1,2\} ~k\in\{2,3\}$, we need, $a_{k1},~a_{k2},~a_{k3},~c_{k1},~c_{k2},~c_{k3} \in \mathbb{R}$, $a_{k4},~c_{k4} \in img(\mathbb{C})$. Also if $p=1,$ $a_{k5},~a_{k6},~a_{k7},~c_{k5},~c_{k6},~c_{k7} \in \mathbb{R}$, $a_{k8}\in img(\mathbb{C}),~c_{k8} \in img(\mathbb{C})$ and if $p=-1$, $a_{k5},~a_{k6},~a_{k7},~c_{k5},~c_{k6},~c_{k7} \in img(\mathbb{C})$, $a_{k8}\in \mathbb{R}, ~c_{k8} \in \mathbb{R}$.

For $A_{1i}$ and $A_{2j}$ ($i,j\in\{2,3\}$) to anti-commute the following conditions need to be satisfied
(we get these conditions by equating $A_{1i}A_{2j}+A_{2j}A_{1i}=0$ and using the linear independence condition over
$\mathbb{C}$).

{\small
\begin{equation}\label{aci4}
\left.
\begin{aligned}
&a_{i2}c_{j4}+pa_{i8}c_{j6}=0;~a_{i2}c_{j5}+a_{i5}c_{j2}=0;~a_{i2}c_{j7}-a_{i3}c_{j6}=0;\\
&a_{i3}c_{j4}+pa_{i8}c_{j7}=0;~a_{i3}c_{j5}+a_{i5}c_{j3}=0;~a_{i4}c_{j2}-pa_{i6}c_{j8}=0;\\
&a_{i4}c_{j3}-pa_{i7}c_{j8}=0;~a_{i4}c_{j6}+a_{i6}c_{j4}=0;~a_{i4}c_{j7}+a_{i7}c_{j4}=0;\\
&a_{i5}c_{j8}+a_{i8}c_{j5}=0;~a_{i6}c_{j3}-a_{i7}c_{j2}=0;~a_{i6}c_{j7}-a_{i7}c_{j6}=0;\\
&~~~~~~~a_{i2}c_{j2}+a_{i3}c_{j3}+pa_{i5}c_{j5}-pa_{i8}c_{j8}=0.
\end{aligned}
\right\}
\end{equation}
}
\textbf{Step 2}:

From Lemma \ref{normal}, by multiplying all the weight matrices by $-A_{11}$ (i.e., $-F_{1}$), we get another equivalent UWD with weight matrices grouped as shown at the top of the next page, after interchanging the first and the second columns.

\begin{figure*}
{\small \begin{center}
 \begin{tabular}{c|c|c|c|c}
 $I_{n}$ & $A_{11}'=-F_{1}$&$A_{21}'=-F_{1}F_{2}$&\ldots & $A_{(2a-2)1}'=-F_{1}F_{2a-2}$\\
 \hline
 $A_{02}'=-F_{1}A_{12}$&$A_{12}'=-F_{1}A_{02}$&$A_{22}'=-F_{1}A_{22}$ & \ldots & $A_{(2a-2)2}'=-F_{1}A_{(2a-2)2}$\\
 \hline
 $A_{03}'=-F_{1}A_{13}$&$A_{13}'=-F_{1}A_{03}$&$A_{23}'=-F_{1}A_{23}$ & \ldots & $A_{(2a-2)3}'=-F_{1}A_{(2a-2)3}$
 \end{tabular}
  \hrule
 \end{center}}
\end{figure*}

It should be noted that in the first row, except $I_{n}$, all are mutually anti-commuting matrices and all of them also anti-commute with $F_{1}F_{2a-1}$ and $F_{1}F_{2a}$. So, -$F_{1}$, -$F_{1}F_{2}$, -$F_{1}F_{3}$, $\cdots$, -$F_{1}F_{2a-1}$ and -$F_{1}F_{2a}$ are $2a$ pairwise anti-commuting matrices.

From (\ref{nascm3}), $A_{1k}'~(k\in\{2,3\})$ has to be anti-commuting with $A_{21}'$, $A_{31}'$, $\cdots$, $A_{(2a-2)1}'$. The only matrices from the set $\{F_{1}^{\lambda_{1}}F_{2}^{\lambda_{2}}\cdots F_{2a}^{\lambda_{2a}}, \lambda_{i}\in \{0,1\} , i=1,2,\cdots,2a\}$ that anti-commute with $F_{1}F_{2}$, $F_{1}F_{3}$, $\cdots$ and $F_{1}F_{2a-2}$ are $F_{1},$ $F_{1}F_{2a-1},$ $F_{1}F_{2a},$ $F_{1}F_{2a-1}F_{2a}$, $F_{2}\cdots F_{2a-2}$, $F_{2}F_{3}\cdots F_{2a-1}$, $F_{2}F_{3}\cdots F_{2a-2}F_{2a}$, $F_{2}F_{3}\cdots F_{2a}$. Hence, let (for $k\in\{2,3\}$)
 \begin{align*}
 A_{1k}' =& b_{k1}F_{1}+b_{k2}F_{1}F_{2a-1}+b_{k3}F_{1}F_{2a}+b_{k4}F_{1}F_{2a-1}F_{2a}\\
 &+b_{k5}F_{2}\cdots F_{2a-2}+b_{k6}F_{2}F_{3}\cdots F_{2a-1}\\
 &+b_{k7}F_{2}F_{3}\cdots F_{2a-2}F_{2a}+b_{k8}F_{2}F_{3}\cdots F_{2a}.
\end{align*}
From (\ref{nascm}), ${(A_{1k}')}^{2}$=$-I_{n}$ and we have

{\small\begin{align*}
{(A_{1k}')}^{2}=&-(b_{k1}^{2}+b_{k2}^{2}+b_{k3}^{2}-b_{k4}^{2}+p(b_{k5}^{2}+b_{k6}^{2}+b_{k7}^{2}-b_{k8}^{2}))I_{n}\\
&-2(b_{k1}b_{k4}+p b_{k5}b_{8})F_{2a-1}F_{2a}\\
&- 2(b_{k2}b_{k4}+p b_{k6}b_{k8})F_{2a}\\
&+2(b_{k1}b_{k6}-b_{k2}b_{k5})F_{1}F_{2}\cdots F_{2a-1}\\
&+2(b_{k1}b_{k7}-b_{k3}b_{k5})F_{1}F_{2}\cdots F_{2a-2}F_{2a}\\
&-2(b_{k2}b_{k7}-b_{k3}b_{k6})F_{1}F_{2}F_{3}\cdots F_{2a}\\
&+2(b_{k3}b_{k4}+p b_{k7}b_{k8})F_{2a-1}.
\end{align*}}
Since, $I_{n},$ $F_{2a-1}F_{2a},$ $F_{1}F_{2}\cdots F_{2a-1}$, $F_{1}F_{2}\cdots F_{2a-2}F_{2a},$ $F_{2a}$,
$F_{1}F_{2}F_{3}\cdots F_{2a},$ $F_{2a-1}$ are linearly independent over $\mathbb{C}$ the following equations have to be satisfied.

{\small\begin{align}
\label{bc1}&b_{k1}^{2}+b_{k2}^{2}+b_{k3}^{2}-b_{k4}^{2}+p(b_{k5}^{2}+b_{k6}^{2}+b_{k7}^{2}-b_{k8}^{2})=1;\\
\nonumber&b_{k1}b_{k4}+pb_{k5}b_{k8}=0;~b_{k1}b_{k6}-b_{k2}b_{k5}=0;~b_{k1}b_{k7}-b_{k3}b_{k5}=0;\\
\nonumber&b_{k2}b_{k4}+pb_{k6}b_{k8}=0;~b_{k2}b_{k7}-b_{k3}b_{k6}=0;~b_{k3}b_{k4}+pb_{k7}b_{k8}=0.
\end{align}}

Since $A_{1k}'^{H}=-A_{1k}'~(k\in\{2,3\})$, we need $b_{k1},~b_{k2},~b_{k3} \in \mathbb{R}$, $b_{k4}\in img(\mathbb{C})$. If $p=1$, $b_{k5},~b_{k6},~b_{k7} \in \mathbb{R}$, $b_{k8} \in img(\mathbb{C})$ and if $p=-1$, $b_{k5},~b_{k6},~b_{k7} \in img(\mathbb{C})$, $b_{k8} \in \mathbb{R}$.

For $A_{1k}'$ and $A_{2j}'$ $(j,k\in\{2,3\})$ to anti-commute the following conditions need to be satisfied
(we get these conditions by equating $A_{1k}'A_{2j}'+A_{2j}'A_{1k}'=0$ and using the linear independence condition over $\mathbb{C}$)

{\small\begin{align}
\nonumber&b_{k2}c_{j4}-pb_{k8}c_{j6}=0;~b_{k2}c_{j5}-b_{k5}c_{j2}=0;~b_{k2}c_{j7}-b_{k3}c_{j6}=0;\\
\nonumber&b_{k3}c_{j4}-pb_{k8}c_{j7}=0;~b_{k3}c_{j5}-b_{k5}c_{j3}=0;~b_{k4}c_{j2}-pb_{k6}c_{j8}=0;\\
\nonumber&b_{k4}c_{j3}-pb_{k7}c_{j8}=0;~b_{k4}c_{j6}+b_{k6}c_{j4}=0;~b_{k4}c_{j7}+b_{k7}c_{j4}=0;\\
\nonumber&b_{k5}c_{j8}+b_{k8}c_{j5}=0;~~b_{k6}c_{j3}-b_{k7}c_{j2}=0;~b_{k6}c_{j7}-b_{k7}c_{j6}=0;\\
\label{bc2}&b_{k2}c_{j2}+b_{k3}c_{j3}-pb_{k5}c_{j5}+pb_{k8}c_{j8}=0.
\end{align}}

Similarly, by equating $A_{1k}'^{H}A_{0i}'+A_{0i}'^{H}A_{1k}'=0$ $(i,k\in\{2,3\})$ and using the linear independence condition over $\mathbb{C}$ we get

{\small\begin{align}
\nonumber&b_{k2}a_{i4}+pb_{k8}a_{i6}=0;~b_{k2}a_{i5}-b_{k5}a_{i2}=0;~b_{k2}a_{i7}-b_{k3}a_{i6}=0;\\
\nonumber&b_{k3}a_{i4}+pb_{k8}a_{i7}=0;~b_{k3}a_{i5}-b_{k5}a_{i3}=0;~b_{k4}a_{i2}-pb_{k6}a_{i8}=0;\\
\nonumber&b_{k4}a_{i3}-pb_{k7}a_{i8}=0;~b_{k4}a_{i6}-b_{k6}a_{i4}=0;~b_{k4}a_{i7}-b_{k7}a_{i4}=0;\\
\nonumber&b_{k5}a_{i8}+b_{k8}a_{i5}=0;~~b_{k6}a_{i3}-b_{k7}a_{i2}=0;~b_{k6}a_{i7}-b_{k7}a_{i6}=0;\\
\label{ab1}&b_{k2}a_{i2}+b_{k3}a_{i3}-pb_{k5}a_{i5}+pb_{k8}a_{i8}=0.
\end{align}}

\textbf{Step 3}:

The conditions (\ref{aci2}) to (\ref{ab1}) can also be re-written as
\begin{align}
\nonumber&\frac{a_{i1}}{a_{i5}}=\frac{-pa_{i8}}{a_{i4}}=\frac{-a_{i2}}{a_{i6}}=\frac{-a_{i3}}{a_{i7}}=z_{i1}\\
\nonumber&\frac{c_{j1}}{c_{j5}}=\frac{-pc_{j8}}{c_{j4}}=\frac{c_{j2}}{c_{j6}}=\frac{c_{j3}}{c_{j7}}=z_{j2}\\
\label{abc}&\frac{b_{k1}}{b_{k5}}=\frac{-pb_{k8}}{b_{k4}}=\frac{b_{k2}}{b_{k6}}=\frac{b_{k3}}{b_{7}}=z_{k9}
\end{align}

The relations between the coefficients in the representation of $A_{1i}$ and $A_{2j},~~i,j \in \{2,3 \}$ are
\begin{align}
\nonumber&\frac{a_{i2}}{pa_{i8}}=\frac{a_{i6}}{a_{i4}}=\frac{-c_{j6}}{c_{j4}}=\frac{c_{j2}}{pc_{j8}}=z_{3};~
\frac{a_{i2}}{a_{i5}}=\frac{-c_{j2}}{c_{j5}}=z_{4};\\
\nonumber&\frac{a_{i2}}{a_{i3}}=\frac{a_{i6}}{a_{i7}}=\frac{c_{j6}}{c_{j7}}=\frac{c_{j2}}{c_{j3}}=z_{5};~
~~~~~\frac{a_{i5}}{pa_{i8}}=\frac{-c_{j5}}{pc_{j8}}=z_{8};\\
\nonumber&\frac{a_{i7}}{a_{i4}}=\frac{a_{i3}}{pa_{i8}}=\frac{-c_{j7}}{c_{j4}}=\frac{c_{j3}}{pc_{j8}}=z_{6};~
\frac{a_{i3}}{a_{i5}}=-\frac{c_{j3}}{c_{j5}}=z_{7};\\
\label{aci5}&a_{i2}c_{j2}+a_{i3}c_{j3}+pa_{i5}c_{j5}-pa_{i8}c_{j8}=0.
\end{align}

The relations between the coefficients in the representation of $A_{1i}$ and $A_{0k}, (i.e., A^\prime_{0i} \mbox{ and }A^\prime_{1k})~~i,k \in \{2,3 \}$ are

\begin{align}
\nonumber&\frac{a_{i2}}{pa_{i8}}=\frac{a_{i6}}{a_{i4}}=\frac{b_{k6}}{b_{k4}}=\frac{-b_{k2}}{pb_{k8}}=z_{10};~
\frac{a_{i2}}{a_{i5}}=\frac{b_{k2}}{b_{k5}}=z_{11};\\
\nonumber&\frac{a_{i2}}{a_{i3}}=\frac{a_{i6}}{a_{i7}}=\frac{b_{k6}}{b_{k7}}=\frac{b_{k2}}{b_{k3}}=z_{12};~
~~~\frac{a_{i5}}{pa_{i8}}=\frac{-b_{k5}}{pb_{k8}}=z_{15};\\
\nonumber&\frac{a_{i7}}{a_{i4}}=\frac{a_{i3}}{pa_{i8}}=\frac{b_{k7}}{b_{k4}}=\frac{-b_{k3}}{pb_{k8}}=z_{13};~
\frac{a_{i3}}{a_{i5}}=\frac{b_{k3}}{b_{k5}}=z_{14};\\
\label{ab3}&a_{i2}b_{k2}+a_{i3}b_{k3}-pa_{i5}b_{k5}+pa_{i8}b_{k8}=0.
\end{align}

The relations between the coefficients in the representation of $A_{0k}$ and $A_{2j}, (i.e., A^\prime_{1k} \mbox{ and }A^\prime_{2j})~~i,k \in \{2,3 \}$ are

\begin{align}
\nonumber&\frac{b_{k2}}{pb_{k8}}=\frac{-b_{k6}}{b_{k4}}=\frac{c_{j6}}{c_{j4}}=\frac{-c_{j2}}{pc_{j8}}=z_{16};~
\frac{b_{k2}}{b_{k5}}=\frac{c_{j2}}{c_{j5}}=z_{17};\\
\nonumber&\frac{b_{k2}}{b_{k3}}=\frac{b_{k6}}{b_{k7}}=\frac{c_{j6}}{c_{j7}}=\frac{c_{j2}}{c_{j3}}=z_{18};~
~~~~~~\frac{b_{k5}}{pb_{k8}}=\frac{-c_{j5}}{pc_{j8}}=z_{21};\\
\nonumber&\frac{-b_{k7}}{b_{k4}}=\frac{b_{k3}}{pb_{k8}}=\frac{c_{j7}}{c_{j4}}=\frac{-c_{j3}}{pc_{j8}}=z_{19};~
\frac{b_{k3}}{b_{k5}}=\frac{c_{j3}}{c_{j5}}=z_{20};\\
\label{bc3}&b_{k2}c_{j2}+b_{k3}c_{j3}-pb_{k5}c_{j5}+pb_{k8}c_{j8}=0.
\end{align}

 Let  ${\bf{a}}_k=$[$a_{k1}$ $a_{k2}$ $a_{k3}$ $a_{k4}$ $a_{k5}$ $a_{k6}$ $a_{k7}$ $ a_{k8}$], ${\bf b}_k=[b_{k1}$ $b_{k2}$ $b_{k3}$ $b_{k4}$ $b_{k5}$ $b_{k6}$ $b_{k7}$ $ b_{k8}$] and $\bf{c}_k=$[$c_{k1}$ $c_{k2}$ $c_{k3}$ $c_{k4}$ $c_{k5}$ $c_{k6}$ $c_{k7}$ $c_{k8}$] for $k\in\{2,3\}$. Let $\bf{a}_1=$[1 0 0 0 0 0 0 0] (according to notation for $F_{1}$ in ${\bf{a}}_k$). Similarly, let $\bf{b}_1=$[1 0 0 0 0 0 0 0] and $\bf{c}_1=$[1 0 0 0 0 0 0 0].
 To construct the weight matrices, it is enough to find linearly independent ${\bf{a}}_k$s, linearly independent ${\bf{b}}_k$s and linearly independent ${\bf{c}}_k$s ($k=1,2,3$), which satisfy conditions (\ref{aci2}) to (\ref{ab1}). Since, any combination of $({\bf{a}}^{(i)},~{\bf{c}}^{(j)},~{\bf{b}}^{(k)}, ~i,~j,~k\in\{1,2\})$ has to satisfy above conditions, $z_{i}~(3\leq i\leq 20,~i\not=9)$ are forced to be constants.

  Now, we find the possibilities under which solution exists for ${\bf{a}}_i,~{\bf{c}}_j,~i,j\in\{2,3\} $.

   From the above conditions assuming all coefficients are non zero, we can write ${\bf{a}}_i$ and ${\bf{c}}_j$ using ${{a}}_{i2}$ and ${{c}}_{j2}$ as (for $i,j\in\{2,3\} $)

{\small
\begin{equation}\label{ao1}
\left.
\begin{aligned}
&a_{i1}=\frac{a_{i2}z_{i1}}{z_{4}};&&a_{i2}=a_{i2};&&a_{i3}=\frac{a_{i2}}{z_{5}};&&a_{i4}=\frac{-a_{i2}}{z_{i1}z_{3}};\\
&a_{i5}=\frac{a_{i2}}{z_{4}};&&a_{i6}=\frac{-a_{i2}}{z_{i1}};&&a_{i7}=\frac{-a_{i2}}{z_{i1}z_{5}};&&pa_{i8}=\frac{a_{i2}}{z_{3}};\\
&c_{j1}=\frac{-c_{j2}z_{j2}}{z_{4}};&&c_{j2}=c_{j2};&&c_{j3}=\frac{c_{j2}}{z_{5}};&&c_{j4}=\frac{-c_{j2}}{z_{j2}z_{3}};\\
&c_{j5}=\frac{-c_{j2}}{z_{4}};&&c_{j6}=\frac{c_{j2}}{z_{j2}};&&c_{j7}=\frac{-c_{j2}}{z_{3}z_{5}};
&&pc_{j8}=\frac{c_{j2}}{z_{3}}
\end{aligned}
\right\}
\end{equation}
}

\begin{equation}\label{x4}
\mbox{and}~~z_{4}=\frac{a_{i2}z_{i1}}{a_{i1}}=\frac{-c_{j2}z_{j2}}{c_{j1}}.
\end{equation}

From (\ref{aci4}), (\ref{ao1}) and (\ref{x4})
\begin{align}
\nonumber a_{i2}c_{j2}\left(1+\frac{1}{z_{5}^{2}} -\frac{p}{z_{3}^{2}}\right)&=\frac{-pa_{i1}c_{j1}}{z_{i1}z_{j2}}.\\
\label{aci6}\text{So,}~~~~~~~ 1+\frac{1}{z_{5}^{2}} -\frac{p}{z_{3}^{2}}&=\frac{-pa_{i1}c_{j1}}{z_{i1}z_{j2}a_{i2}c_{j2}}
=\frac{pc_{j1}^{2}}{c_{j2}^{2}z_{j2}^{2}}=\frac{pa_{i1}^{2}}{a_{i2}^{2}z_{i1}^{2}}.
\end{align}
From (\ref{aci1}) and (\ref{ao1}) to (\ref{aci6})
{\small\begin{align}
\nonumber 1&=\left(1+\frac{p}{z_{i1}^{2}}\right)(a_{i1}^{2}+a_{i2}^{2}+a_{i3}^{2}-pa_{i8}^{2})\\
\nonumber&=\left(1+\frac{p}{z_{i1}^{2}}\right)\left(a_{i1}^{2}+a_{i2}^{2}\left(1+\frac{1}{z_{5}^{2}} -\frac{p}{z_{3}^{2}}\right)\right)
=\left(1+\frac{p}{z_{i1}^{2}}\right)^{2}a_{i1}^{2}\\
\label{aci7}&\mbox{and similarly}~\left(1+\frac{p}{z_{j2}^{2}}\right)^{2}c_{j1}^{2}=1~ .
\end{align}}

For a given $p$, by choosing the values of $z_{i1}, ~z_{j2}, ~z_{3}, ~z_{5}, ~a_{i1}, ~a_{i2}, ~c_{j1}, ~c_{j2}$ we can get ${\bf a}_i$ and ${\bf c}_j$. For example, for $p=1$, let $z_{3}=\frac{i}{2}$, $z_{4}=\frac{-1}{3}$, $z_{5}=\frac{1}{2}$, $c_{2}=-a_{2}=\frac{1}{6}$. Now, for $z_{21}=1=z_{22}$ and $a_{21}=c_{21}=\frac{1}{2}$ we get
\begin{align*}
{\bf a}_{2}&=\left[\frac{1}{2}~\frac{-1}{6}~\frac{-1}{3}~\frac{-i}{3}~\frac{1}{2}~\frac{1}{6}~\frac{1}{3}~\frac{i}{3}\right];\\
{\bf c}_{2}&=\left[\frac{1}{2}~\frac{1}{6}~\frac{1}{3}~\frac{i}{3}~\frac{1}{2}~\frac{1}{6}~\frac{1}{3}~\frac{-i}{3}\right],
\end{align*}
and, for $z_{31}=-1=z_{32}$ and $a_{31}=c_{31}=\frac{-1}{2}$, we get
\begin{align*}
{\bf a}_{3}&=\left[\frac{-1}{2}~\frac{-1}{6}~\frac{-1}{3}~\frac{i}{3}~\frac{1}{2}~\frac{-1}{6}~\frac{-1}{3}~\frac{i}{3}\right];\\
{\bf c}_{3}&=\left[\frac{-1}{2}~\frac{1}{6}~\frac{1}{3}~\frac{-i}{3}~\frac{1}{2}~\frac{-1}{6}~\frac{-1}{3}~\frac{-i}{3}\right].
\end{align*}

Now we consider the cases where some of the coefficients $a_{ik}$s or $c_{jk}$s are zero. For these cases we use (\ref{aci1}) to (\ref{aci7}).

\textbf{Case 1}: Let $a_{21}=0$. Then, $a_{25}=0$ or $a_{22}=a_{23}=a_{28}=0$.

Let $a_{22}=a_{23}=a_{28}=0$ and $a_{25}\not=0$. Then, we get $c_{j2}=c_{j3}=c_{j8}=c_{j5}=0,~j=2,3$. So, $c_{j1}c_{j4}=c_{j1}c_{j6}=c_{j1}c_{j7}=0$. If $c_{j1}\not=0$, $A_{2j}=\pm F_{2}$ is clearly not a solution. This implies $c_{1}=0$. From (\ref{aci4}), we have (for $i,j\in\{2,3\}$)
\begin{equation*}
\frac{a_{i4}}{a_{i6}}=\frac{-c_{j4}}{c_{j6}}=y_{1};~~
\frac{a_{i4}}{a_{i7}}=\frac{-c_{j4}}{c_{j7}}=y_{2};~~
\frac{a_{i6}}{a_{i7}}=\frac{c_{j6}}{c_{j7}}=y_{3},
\end{equation*}
for some constants $y_{1},~y_{2},~y_{3}.$  From these equations, ${\bf c}_{2}$ and ${\bf c}_{3}$ are linearly dependent. So, assuming $a_{22}=a_{23}=a_{28}=0$ and $a_{25}\not=0$ is not valid.

Now, let $a_{22}=a_{23}=a_{28}=0$ and $a_{25}=0$. Consider $z_4,z_7,z_8,z_{11},z_{14},z_{15},z_{17},z_{20},z_{21}$.
Let $a_{35}\not=0$. Now, if  $a_{32}\not=0$, then $c_{j2}=c_{j5}=0$ or $b_{j2}=b_{j5}=0$, $j=2,3$. With out loss of generality assume $c_{j2}=c_{j5}=0$, $j=2,3$. Then, $c_{j3}=c_{j8}=0$, $j=2,3$. Then $c_j$s, $j=1,2,3$ cannot be linearly independent. This happens even if $a_{33}\not=0$ or $a_{38}\not=0$. So $a_{32}=a_{33}=a_{38}=0$, when  $a_{35}\not=0$. Then, again $c_{j2}=c_{j3}=c_{j5}=c_{j8}=0$, $j=2,3$, which is not valid. So $a_{35}=0$. If $a_{32}=a_{33}=a_{38}=0$, then for ${\bf a}_{1}$, ${\bf a}_{2}$ and ${\bf a}_{3}$ to be linearly independent, we need to have $c_{j2}=c_{j3}=c_{j4}=c_{j6}=c_{j7}=c_{j8}=0$, $j=2,3$. In that case, ${\bf c}_{1}$, ${\bf c}_{2}$ and ${\bf c}_{3}$ cannot satisfy linear independence conditions (because only $c_{j1}\not=0$ and $c_{j5}\not=0$). So, with out loss of generality let $a_{32}\not=0$. Then $c_{j5}=0$,
$j=2,3$. Since $a_{35}=0$, $a_{31}=0$ or $a_{34}=a_{36}=a_{37}=0$. Let $a_{34}=a_{36}=a_{37}=0$. Since $c_{j5}=0$, $c_{j1}=0$ or $c_{j4}=c_{j6}=c_{j7}=0$, $j=2,3$. Let $c_{j4}=c_{j6}=c_{j7}=0$, $j=2,3$. Then, for $c_{j}$s, $j=1,2,3$ to be linearly independent $a_{24}=a_{26}=a_{27}=0$, which is not valid. So, let $c_{24}=c_{26}=c_{27}=0$ and $c_{31}=0$. From \ref{aci5}, if $a_{32}\not=0$, $a_{33}\not=0$ and $a_{38}\not=0$, then, $c_{22}=c_{23}=c_{28}=0$ which is not possible. If $a_{38}=0$, then $c_{34}=0$ , then
$a_{24}=0$, then $c_{28}=c_{38}=0$, which gives $c_{22}^{2}+c_{23}^{2}=0$. So $c_{22}=c_{23}=0$, because $c_{22},c_{23}\in \mathbb{R}$. So, let $c_{j1}=0$, $j=2,3$. Now, as in previous assumption $c_{j2}=c_{j3}=c_{j8}=0$, $j=2,3$. Now, only $c_{j4},c_{j6},c_{j7}$ are possibly non-zero. So, we cannot get
$a_{i}$s and $c_{i}$s satisfying linear independence conditions. So, consider $a_{31}=0$. Again, since $c_{j5}=0$, $c_{j1}=0$ or $c_{j4}=c_{j6}=c_{j7}=0$, $j=2,3$. If $c_{j4}=c_{j6}=c_{j7}=0$, $j=2,3$, we cannot get $a_{i}$s and $c_{i}$s satisfying linear independence conditions. So, let $c_{24}=c_{26}=c_{27}=0$ and $c_{31}=0$.  From \ref{aci5}, if $a_{32}\not=0$, $a_{33}\not=0$ and $a_{38}\not=0$, then, $c_{22}=c_{23}=c_{28}=0$ which is not possible.  If $a_{38}=0$, since $a_{32}\not=0$ then $c_{34}=0$ , then
$a_{24}=a_{34}=0$, then $c_{28}=c_{38}=0$. If $a_{32}\not=0$ and $a_{33}\not=0$, then $c_{22}=c_{23}=0$, so not a solution. So $a_{33}=0$. Then, $c_{22}=c_{32}=0$, $c_{37}=0$, $a_{37}=0$, $a_{27}=0$ and $c_{23}=c_{33}=0$, which is not a solution. So, let $c_{j1}=0$, $j=2,3$. Now, as in the previous assumption $a_{32}\not=0$, $a_{33}\not=0$ and $a_{38}\not=0$ is not possible and $a_{32}\not=0$, $a_{33}\not=0$ and $a_{38}=0$ is also not possible. So, let $a_{32}\not=0$, $a_{33}=0$ and $a_{38}=0$. Then $c_{j4}=0$, $a_{j4}=0$ and $c_{j8}=0$, $j=2,3$. And, $c_{j2}=0$, $c_{j7}=0$, $a_{j7}=0$ and $c_{j3}c_{j6}=0$, $j=2,3$. So, this is not a valid option.

Now, let $a_{25}=0$. Then, we get $a_{22}c_{j5}=a_{23}c_{j5}=a_{28}c_{j5}=0$, $j=2,3$. Since, for $a_{22}=a_{23}=a_{28}=a_{25}=0$, (\ref{aci1}) to (\ref{aci7}) do not have a solution, $c_{j5} =0$, $j=2,3$. Now, $c_{j1}=0$ or $c_{j4}=c_{j6}=c_{j7}=0$ has to be satisfied. Let $c_{j4}=c_{j6}=c_{j7}=0$, $j=2,3$. Then, we have (for $i,j\in\{2,3\}$)
\begin{equation*}
\frac{c_{j2}}{pc_{j8}}=\frac{a_{i6}}{a_{i4}}=y_4;~~
\frac{c_{j3}}{pc_{j8}}=\frac{a_{i7}}{a_{i4}}=y_5;~~
\frac{c_{j2}}{c_{j3}}=\frac{a_{i6}}{a_{i7}}=y_6;
\end{equation*}
\begin{equation}\label{innerprod}
a_{i2}c_{j2}+a_{i3}c_{j3}-pa_{i8}c_{j8}=0,
\end{equation}
for some constants $y_{4},~y_{5},~y_{6}.$ For ${\bf c}_{j}$s to be linearly independent, we need $a_{i4}=a_{i6}=a_{i7}=0=a_{35}$. Since ${\bf c}_{j}$s and ${\bf a}_{i}$s have to satisfy linearly independence conditions, we need 4 vectors [$a_{22}$ $a_{23}$ $ a_{28}$], [$a_{32}$ $a_{33}$ $a_{38}$], [$c_{22}$ $c_{23}$ $ c_{28}$], [$c_{32}$ $c_{33}$ $ c_{38}$] such that $a_{i2}c_{j2}+a_{i3}c_{j3}-pa_{i8}c_{j8}=0$ ($i,j\in\{2,3\}$).

Since, $a_{i4}=a_{i5}=a_{i6}=a_{i7}=0=c_{j4}=c_{j5}=c_{j6}=c_{j7} $ ($i,j\in\{2,3\}$) and $A_{1i}s, A_{2j}s$  ($i,j\in\{1,2,3\}$) (from \ref{b_rep}) are to be linearly independent, [$a_{22}$ $a_{23}$ $ a_{28}$], [$a_{32}$ $a_{33}$ $a_{38}$], [$c_{22}$ $c_{23}$ $ c_{28}$] and [$c_{32}$ $c_{33}$ $ c_{38}$] have to be linearly independent over $\mathbb{R}$, which is not possible. Hence, $c_{j4}=c_{j6}=c_{j7}=0$, $j=2,3$ is not a valid option.

Now, let $c_{24}=c_{26}=c_{27}=0$ and $c_{31}=0$. No solution if $a_{i2}=a_{i3}=a_{i8}=0$ for some $i\in\{2,3\}$. Let $a_{i2}\not=0$, $a_{i3}\not=0$ and $a_{i8}\not=0$, $i\in\{2,3\}$. Then $c_{34}\not=0$, $c_{36}\not=0$ and $c_{37}\not=0$. If $a_{i4}=a_{i6}=a_{i7}=0$ for both $i\in\{2,3\}$, $a_i$s won't be linearly independent. So, let $a_{k4}\not=0$, $a_{k6}\not=0$ and $a_{k7}\not=0$ for some $k\in\{2,3\}$. Then from (\ref{innerprod}) $c_{j2}=c_{j3}=c_{j8}=0$, $j=2,3$, so not valid. Now, let $a_{28}=0$. Then $c_{34}=0$, $a_{i4}=0$, $c_{i8}=0$, $i\in\{2,3\}$ and $a_{38}=0$. If $a_{22}\not=0$ and $a_{23}\not=0$, $c_{j2}=c_{j3}=0$, $j=2,3$, so not a solution. So, let $a_{23}=0$. Then $c_{37}=0$, $a_{i7}=0$ and $c_{j2}=c_{j3}=0$, $j=2,3$, so not a solution.

Now, let $c_{21}=c_{31}=0$. By proceeding as in above assumption, we get $c_{i4}=0$, $a_{i4}=0$, $c_{i8}=0$, $i\in\{2,3\}$ and $a_{38}=0$, when $a_{28}=0$. If $a_{22}\not=0$ and $a_{23}\not=0$, $c_{j2}=c_{j3}=0$, $j=2,3$, so not a solution. So, let $a_{23}=0$. Then $c_{i7}=0$, $a_{i7}=0$ and $c_{j2}=c_{j3}=0$, $j=2,3$, so $a_{21}=a_{31}=0$ not valid.

\textbf{Case 2}: Let $a_{25}=0$. Then, $a_{21}=0$ or $a_{24}=a_{26}=a_{27}=0$. Since $a_{21}=0$ is not possible, $a_{24}=a_{25}=a_{26}=a_{27}=0$. Now, we get $a_{23}c_{j5}=a_{28}c_{j5}=a_{22}c_{j5}=0$, $j\in\{2,3\}$. Since $a_{22}=a_{23}=a_{24}=a_{25}=a_{26}=a_{27}=a_{28}=0$ is not a valid solution, $c_{j5}=0$, $j\in\{2,3\}$. Then, $c_{j1}=0$ or $c_{j4}=c_{j6}=c_{j7}=0$. Since $c_{j1}=0$ is not valid (from Case 1, because of similarity between ${\bf a}_i$ and ${\bf c}_j$), $c_{j4}=c_{j5}=c_{j6}=c_{j7}=0$. Similarly, since $c_{j5}=0$, $a_{35}=0$ and $a_{34}=a_{36}=a_{37}=0$. So, from Case 1, we cannot get $a_{i}$s and $c_{i}$s satisfying linear independence conditions. So, $a_{25}=0$ is not possible.

\textbf{Case 3}: Let $a_{22}=0$. Then, $a_{26}=0$ or $a_{21}=a_{23}=a_{28}=0$. From Case 1, $a_{21}=0$ is not valid. So, $a_{26}=0$. Now, $c_{j6}a_{28}=c_{j6}a_{23}=c_{j2}a_{25}=0$, $j\in\{2,3\}$. Since $a_{25}=0$ is not valid, $c_{j2}=0$ ($j\in\{2,3\}$). Since $c_{j1}=c_{j3}=c_{j8}=0$ is not valid, $c_{j6}=0$ ($j\in\{2,3\}$). Now, $a_{i2}=a_{i6}=c_{j2}=c_{j6}=0$, $i,j\in\{2,3\}$. If we denote ${\bf a}_i$ and ${\bf c}_j$ using $a_{i3}$ and $c_{j3}$, instead of $a_{i2}$ and $c_{j2}$, and follow the lines of (\ref{aci6}) and (\ref{aci7}), we get ${\bf a}_{i}$s and ${\bf c}_{j}$s, satisfying the linear independence conditions.

\textbf{Case 4}: Let $a_{23}=0$. Then, $a_{27}=0$ or $a_{21}=a_{22}=a_{28}=0$. From Case 1, $a_{21}=0$ is not valid. So, $a_{27}=0$. Then, we get $a_{22}c_{j7}=a_{28}c_{j7}=a_{25}c_{j3}=a_{24}c_{j3}=a_{24}c_{j7}=a_{26}c_{j3}=a_{26}c_{j7}=0$, $j\in\{2,3\}$.
 If $c_{j3}\not=0$, then $a_{25}=a_{24}=a_{26}=a_{27}=0,$ and this is not possible from Case 2. So, $c_{j3}=0$. Since $c_{j1}=c_{j2}=c_{j8}=0$ is not valid, $c_{j7}=0$ ($j\in\{2,3\}$). Then, $a_{i3}=a_{i7}=c_{j3}=c_{j7}=0$. By following the lines of (\ref{aci6}) and (\ref{aci7}), we get ${\bf a}_{i}$s and ${\bf c}_{j}$s, satisfying the linear independence conditions.

\textbf{Case 5}: Let $a_{24}=0$. Then, $a_{28}=0$, because $a_{24}=a_{25}=a_{26}=a_{27}=0$ is not valid from Case 2. Now, we get $a_{22}c_{j4}=a_{23}c_{j4}=a_{26}c_{j8}=a_{27}c_{j8}=a_{26}c_{j4}=a_{27}c_{j4}=a_{25}c_{j8}=0$, $j\in\{2,3\}$.
  If $c_{j8}\not=0$, $a_{25}=a_{26}=a_{27}=a_{24}=a_{28}=0$, and this is not possible from Case 2. So, $c_{j8}=0$. Since $c_{j1}=c_{j2}=c_{j3}=0$ is not valid, $c_{j4}=0$ ($j\in\{2,3\}$).
Now, $a_{i4}=a_{i8}=c_{j4}=c_{j8}=0$. By following the lines of (\ref{aci6}) and (\ref{aci7}), we get ${\bf a}_{i}$s and ${\bf c}_{j}$s, satisfying the linear independence conditions.

\textbf{Case 6} : Let $a_{26}=0$. Since, $a_{24}=a_{25}=a_{26}=a_{27}=0$ is not valid from Case 2, $a_{22}=0$. So, Case 3 follows.

\textbf{Case 7} : Let $a_{27}=0$. Since, $a_{24}=a_{25}=a_{26}=a_{27}=0$ is not valid from Case 2, $a_{23}=0$. So Case 4 follows.

\textbf{Case 8} : Let $a_{28}=0$. Then, $a_{24}=0$ or $a_{21}=a_{22}=a_{23}=0$. From Case 1, $a_{21}=a_{23}=a_{28}=a_{22}=0$ is not valid. So, $a_{24}=0$ and Case 5 follows.

From the above cases, it follows that  ${\bf a}_i$s, and ${\bf c}_i$s ($i\in\{2,3\}$), satisfying the linear independence conditions are possible only if
\begin{enumerate}
  \item None of the coefficients in ${\bf a}_i$ or ${\bf c}_i$ are 0.
  \item $a_{i2}=a_{i6}=c_{i2}=c_{i6}=0$ and other coefficients are non-zero.
  \item $a_{i3}=a_{i7}=c_{i3}=c_{i7}=0$ and other coefficients are non-zero.
  \item $a_{i4}=a_{i8}=c_{i4}=c_{i8}=0$ and other coefficients are non-zero.
  \item $a_{i2}=a_{i6}=c_{i2}=c_{i6}=a_{i3}=a_{i7}=c_{i3}=c_{i7}=0$ and other coefficients are non-zero.
  \item $a_{i2}=a_{i6}=c_{i2}=c_{i6}=a_{i4}=a_{i8}=c_{i4}=c_{i8}=0$ and other coefficients are non-zero.
  \item $a_{i3}=a_{i7}=c_{i3}=c_{i7}=a_{i4}=a_{i8}=c_{i4}=c_{i8}=0$ and other coefficients are non-zero.
\end{enumerate}

Here, $a_{i2}=a_{i6}=c_{i2}=c_{i6}=a_{i3}=a_{i7}=c_{i3}=c_{i7}=a_{i4}=a_{i8}=c_{i4}=c_{i8}=0$ and other coefficients are non-zero, possibility is not taken into account, because, we are left with only $a_{i1},~a_{i5},~c_{i1},~c_{i5}$ as non-zeros. From which we cannot get ${\bf a}_i$s, and ${\bf c}_i$s ($i\in\{2,3\}$) satisfying the linear independence conditions.

\textbf{Step 4}:

As in the case of solving for (${\bf a}_i$ and ${\bf c}_j$), for solving (${\bf a}_i$ and ${\bf b}_k$) or (${\bf b}_k$ and ${\bf c}_j$) we get similar conditions on their coefficients. Therefore, we may be able to get solutions to ${\bf a}_i$s, ${\bf b}_k$s and ${\bf c}_j$s, which satisfy the linear independence conditions, if they satisfy one of the following conditions: (for $i\in\{2,3\}$)
\begin{enumerate}
  \item None of the coefficients in ${\bf a}_i$ or ${\bf b}_i$ or ${\bf c}_i$ are 0
  \item $a_{i2}=a_{i6}=b_{i2}=b_{i6}=c_{i2}=c_{i6}=0$ and other coefficients are non-zero.
  \item $a_{i3}=a_{i7}=b_{i3}=b_{i7}=c_{i3}=c_{i7}=0$ and other coefficients are non-zero.
  \item $a_{i4}=a_{i8}=b_{i4}=b_{i8}=c_{i4}=c_{i8}=0$ and other coefficients are non-zero.
  \item $a_{i2}=a_{i6}=b_{i2}=b_{i6}=c_{i2}=c_{i6}=a_{i3}=a_{i7}=b_{i3}=b_{i7}=c_{i3}=c_{i7}=0$ and other coefficients are non-zero.
  \item $a_{i2}=a_{i6}=b_{i2}=b_{i6}=c_{i2}=c_{i6}=a_{i4}=a_{i8}=b_{i4}=b_{i8}=c_{i4}=c_{i8}=0$ and other coefficients are non-zero.
  \item $a_{i3}=a_{i7}=b_{i3}=b_{i7}=c_{i3}=c_{i7}=a_{i4}=a_{i8}=b_{i4}=b_{i8}=c_{i4}=c_{i8}=0$ and other coefficients are non-zero.
\end{enumerate}

\textbf{Step 5}:

From (\ref{aci5}), (\ref{ab3}) and (\ref{bc3}), since, $a_{i5}$, $b_{k5}$, $c_{j5}$ cannot be zero, we have (for $i,j,k\in\{2,3\}$)
\begin{align*}
&a_{i2}=a_{i6}=b_{k2}=b_{k6}=c_{j2}=c_{j6}=0 ~~\text{from}~~z_{4},z_{11},z_{17},\\
&a_{i3}=a_{i7}=b_{k3}=b_{k7}=c_{j3}=c_{j7}=0 ~~\text{from}~~z_{7},z_{14},z_{20},\\
&a_{i4}=a_{i8}=b_{k4}=b_{k8}=c_{j4}=c_{j8}=0 ~~\text{from}~~z_{8},z_{15},z_{21}.
\end{align*}

Now, we are left with only $a_{i1}$, $a_{i5}$, $b_{k1}$, $b_{k5}$, $c_{j1}$, $c_{j5},$  which are non-zero, from which we cannot get ${\bf a}_i$s,  ${\bf b}_k$s and ${\bf c}_j$s which satisfy the linear independence conditions. Hence, $g= 2a-1$ is not possible.

From Corollary \ref{corcode}, since we are able to generate $g=2a-2$ group 3-real symbol decodable UWDs, the maximum achievable rate is $\frac{3(2a-2)}{2^{a+1}}=\frac{3(a-1)}{2^{a}}$ cspcu. Since this rate is achievable, this upper bound is tight. This completes the proof.

\section*{Appendix C}
\begin{center}
{\bf Proof of Theorem \ref{thm4}}
\end{center}
\begin{proof}
Let $g$ be the number of groups. Then, the rate is given by $\frac{4g}{2^{a+1}}=\frac{g}{2^{a-1}}$ cspcu. From Theorem \ref{thm3sym2}, the maximum number of groups possible for 3-real symbol decodable UWDs is $2a-2$. So, for 4-real symbol decodable UWDs the maximum number of groups possible is $\le 2a-2$, but in Theorem \ref{thmcode}, we constructed $2a-2$ group 4-real symbol decodable UWDs. So, the maximum number of groups possible (achievable) for 4-real symbol decodable UWDs is $2a-2$. Hence, the tight upper bound on the rate of 4-real symbol decodable $2^{a}\times 2^{a}$ ($a\ge2$) UWD is $\frac{a-1}{2^{a-2}}$ cspcu.
\end{proof}


\section*{Appendix D}
\begin{center}
{\bf Proof of Theorem \ref{div_thm}}
\end{center}

Let $S$ and ${S}'$ be two distinct codewords of the code as in Theorem \ref{thmcode}. Let
{\small \begin{align*}
    {S}=\sum_{i=0}^{2a-3}\sum_{j=1}^{4}x_{ij}{A}_{ij},~~~
    {S}'=\sum_{i=0}^{2a-3}\sum_{j=1}^{4}x_{ij}'{A}_{ij},
  \end{align*}}
where $A_{01}=I_{n}$. For $0\leq i \leq 2a-3$ and $1\leq j_1,j_2 \leq 4$,
\begin{equation}\label{dcg1}
A_{i{j_1}}^{H}A_{i{j_2}}+A_{i{j_2}}^{H}A_{i{j_1}}=2A_{0{j_1}}A_{0{j_2}}.
\end{equation}

Let $\bigtriangleup {S}\triangleq {S}-{S}'$, $(\bigtriangleup {\bf x}_{i}{\bf x}_{i}')_{j}\triangleq x_{ij}-x_{ij}'$ and $\bigtriangleup {\bf x}_{i}{\bf x}_{i}'\triangleq [(\bigtriangleup {\bf x}_{i}{\bf x}_{i}')_{1}~(\bigtriangleup {\bf x}_{i}{\bf x}_{i}')_{2}~(\bigtriangleup {\bf x}_{i}{\bf x}_{i}')_{3}~(\bigtriangleup {\bf x}_{i}{\bf x}_{i}')_{4}]^{T}$. Then, $(\bigtriangleup {S})^{H}(\bigtriangleup {S})$ is given by
{\small\begin{align}
\nonumber
&\left[\sum_{i=0}^{2a-3}\sum_{j=1}^{4}(\bigtriangleup {\bf x}_{i}{\bf x}_{i}')_{j}{A}_{ij}\right]^{H}
\left[\sum_{m=0}^{2a-3}\sum_{j=1}^{4}(\bigtriangleup {\bf x}_{i}{\bf x}_{i}')_{j}{A}_{mj}\right]\\
\label{shsdet}=&\sum_{i=0}^{2a-3}\left(\sum_{j=1}^{4}(\bigtriangleup {\bf x}_{i}{\bf x}_{i}')_{j}{A}_{ij}\right)^{H}\left(\sum_{j=1}^{4}(\bigtriangleup {\bf x}_{i}{\bf x}_{i}')_{j}{A}_{ij}\right)\\
\nonumber=&\sum_{i=0}^{2a-3}[((\bigtriangleup {\bf x}_{i}{\bf x}_{i}')_{1}^{2}+(\bigtriangleup {\bf x}_{i}{\bf x}_{i}')_{2}^{2}+(\bigtriangleup {\bf x}_{i}{\bf x}_{i}')_{3}^{2}+(\bigtriangleup {\bf x}_{i}{\bf x}_{i}')_{4}^{2})I_{n}+\\
\nonumber &2(\bigtriangleup {\bf x}_{i}{\bf x}_{i}')_{1}(\bigtriangleup {\bf x}_{i}{\bf x}_{i}')_{2}A_{02}+2(\bigtriangleup {\bf x}_{i}{\bf x}_{i}')_{1}(\bigtriangleup {\bf x}_{i}{\bf x}_{i}')_{3}A_{03}+\\
\nonumber &2(\bigtriangleup {\bf x}_{i}{\bf x}_{i}')_{1}(\bigtriangleup {\bf x}_{i}{\bf x}_{i}')_{4}A_{04}
+2(\bigtriangleup {\bf x}_{i}{\bf x}_{i}')_{2}(\bigtriangleup {\bf x}_{i}{\bf x}_{i}')_{3}A_{02}A_{03}+\\
\label{shsdet1}&2(\bigtriangleup {\bf x}_{i}{\bf x}_{i}')_{2}(\bigtriangleup {\bf x}_{i}{\bf x}_{i}')_{4}A_{02}A_{04}+2(\bigtriangleup {\bf x}_{i}{\bf x}_{i}')_{3}(\bigtriangleup {\bf x}_{i}{\bf x}_{i}')_{4}A_{03}A_{04}],
\end{align}}
where, (\ref{shsdet}) and (\ref{shsdet1}) occurs because of (\ref{nasc}), (\ref{nasc1}) and (\ref{dcg1}).

Now, we calculate the determinant of $(\bigtriangleup {S})^{H}(\bigtriangleup {S})$ using $A_{02},~A_{03},~A_{04}$ as given in Theorem \ref{thmcode}.
{\small\begin{align}
\nonumber &det[(\bigtriangleup {S})^{H}\bigtriangleup {S}]
=\\ \nonumber&\left[\sum_{i=0}^{2a-3}((\bigtriangleup {\bf x}_{i}{\bf x}_{i}')_{1}+(\bigtriangleup {\bf x}_{i}{\bf x}_{i}')_{2}+(\bigtriangleup {\bf x}_{i}{\bf x}_{i}')_{3}+(\bigtriangleup {\bf x}_{i}{\bf x}_{i}')_{4})^{2}\right]^{\frac{n}{4}}\\
\nonumber&\left[\sum_{i=0}^{2a-3}((\bigtriangleup {\bf x}_{i}{\bf x}_{i}')_{1}-(\bigtriangleup {\bf x}_{i}{\bf x}_{i}')_{2}+(\bigtriangleup {\bf x}_{i}{\bf x}_{i}')_{3}-(\bigtriangleup {\bf x}_{i}{\bf x}_{i}')_{4})^{2}\right]^{\frac{n}{4}}\\
\nonumber&\left[\sum_{i=0}^{2a-3}((\bigtriangleup {\bf x}_{i}{\bf x}_{i}')_{1}+(\bigtriangleup {\bf x}_{i}{\bf x}_{i}')_{2}-(\bigtriangleup {\bf x}_{i}{\bf x}_{i}')_{3}-(\bigtriangleup {\bf x}_{i}{\bf x}_{i}')_{4})^{2}\right]^{\frac{n}{4}}\\
\label{divdet}&\left[\sum_{i=0}^{2a-3}((\bigtriangleup {\bf x}_{i}{\bf x}_{i}')_{1}-(\bigtriangleup {\bf x}_{i}{\bf x}_{i}')_{2}-(\bigtriangleup {\bf x}_{i}{\bf x}_{i}')_{3}+(\bigtriangleup {\bf x}_{i}{\bf x}_{i}')_{4})^{2}\right]^{\frac{n}{4}}
\end{align}}

The minimum of the determinant, denoted by $\bigtriangleup_{min},$ of $(\bigtriangleup {S})^{H}(\bigtriangleup {S})$ for all possible non-zero $\bigtriangleup {S}$ is given as
\begin{equation}
\label{divdet1}\bigtriangleup_{min}=\min_{\bigtriangleup {S}\not= 0} det[(\bigtriangleup {S})^{H}(\bigtriangleup {S})].
\end{equation}
Since the expression in the right hand side of equation (\ref{divdet}) is a product of sum of squares of real numbers, we can write (\ref{divdet}), (\ref{divdet1}) as

{\small
\begin{align*}
\nonumber &det[(\bigtriangleup {S})^{H}(\bigtriangleup {S})]
\geq\\ \nonumber&\left[((\bigtriangleup {\bf x}_{i}{\bf x}_{i}')_{1}+(\bigtriangleup {\bf x}_{i}{\bf x}_{i}')_{2}+(\bigtriangleup {\bf x}_{i}{\bf x}_{i}')_{3}+(\bigtriangleup {\bf x}_{i}{\bf x}_{i}')_{4})^{2}\right]^{\frac{n}{4}}\\
&\left[((\bigtriangleup {\bf x}_{i}{\bf x}_{i}')_{1}-(\bigtriangleup {\bf x}_{i}{\bf x}_{i}')_{2}+(\bigtriangleup {\bf x}_{i}{\bf x}_{i}')_{3}-(\bigtriangleup {\bf x}_{i}{\bf x}_{i}')_{4})^{2}\right]^{\frac{n}{4}}\\
&\left[((\bigtriangleup {\bf x}_{i}{\bf x}_{i}')_{1}+(\bigtriangleup {\bf x}_{i}{\bf x}_{i}')_{2}-(\bigtriangleup {\bf x}_{i}{\bf x}_{i}')_{3}-(\bigtriangleup {\bf x}_{i}{\bf x}_{i}')_{4})^{2}\right]^{\frac{n}{4}}\\
&\left[((\bigtriangleup {\bf x}_{i}{\bf x}_{i}')_{1}-(\bigtriangleup {\bf x}_{i}{\bf x}_{i}')_{2}-(\bigtriangleup {\bf x}_{i}{\bf x}_{i}')_{3}+(\bigtriangleup {\bf x}_{i}{\bf x}_{i}')_{4})^{2}\right]^{\frac{n}{4}}
\end{align*}}
for some $0\leq i \leq 2a-3$. And, $\bigtriangleup_{min}$ occurs when all but one among $\bigtriangleup {\bf x}_{i}{\bf x}_{i}',~0\leq i \leq 2a-3$ are zeros.

Therefore,
$ \bigtriangleup_{min}=\min_{\bigtriangleup {\bf x}_{i}{\bf x}_{i}' \not= 0} (\Psi(\bigtriangleup {\bf x}_{i}{\bf x}_{i}'))^{n},$
for $ \bigtriangleup {\bf x}_{i}{\bf x}_{i}'\in \bigtriangleup \cal{B}$ (i.e. ${\bf x}_{i}\in \cal{B}$) $\forall~ 0\leq i \leq 2a-3$.
 To achieve full diversity we need $\bigtriangleup_{min}>0$, which can be guaranteed if  (\ref{signal_condn}) is satisfied. This completes the proof.


\begin{thebibliography}{160}
\bibitem{KSR}
Sanjay Karmaker, K. Pavan Srinath and B. Sunder Rajan, ``Maximum Rate of Unitary-Weight, Single-Symbol Decodable STBCs'', To appear in IEEE Transactions on Information Theory. Available as arXiv:1101.2516v1 [cs.IT] 13 Jan 2011.

\bibitem{RaR}
G. Susinder Rajan and B. Sundar Rajan, ``Multigroup ML Decodable Collocated and Distributed Space-Time Block Codes'', IEEE Transactions on Information Theory, vol.56, no.7, pp. 3221-3247, July 2010.

\bibitem{TSC}
\ V.Tarokh, N.Seshadri and A.R Calderbank,"Space time codes for high date rate wireless communication : performance criterion and code construction'',
\emph{IEEE Trans. inform theory.}, vol.\ 44,  pp. \  744 - 765, 1998.

\bibitem{KhR}
Md. Zafar Ali Khan and B. Sundar Rajan, "Single-Symbol Maximum-Likelihood Decodable Linear STBCs," IEEE Trans. Inf. Theory, vol. 52, no. 5, pp. 2062-2091, May 2006.

\bibitem{KaR}
Sanjay Karmakar, B.Sundar Rajan, "Minimum-Decoding Complexity, Maximum-rate Space-Time Block Codes from
Clifford Algebras," Proc. IEEE Intnl. Symp. Inform. Theory, Seattle, July 9-14, 2006, pp. 788-792.

\bibitem{KaR2}
Sanjay Karmakar and B. Sundar Rajan, "High-rate Double-Symbol-Decodable STBCs from Clifford Algebras," Proc. IEEE Conf. Global Commun., San Francisco, CA, Nov. 27-Dec. 1, 2006.


\bibitem{RaR2}
G. Susinder Rajan and B. Sundar Rajan, "Algebraic Distributed Differential Space-Time Codes with Low Decoding
Complexity," IEEE Trans. Wireless Commun., vol. 7, no. 10, pp. 3962-3971, Oct. 2008.


\bibitem{ShM}
Daniel B.Shapiro and Reiner Martin, ``Anticommuting Matrices'', The American Mathematical Monthly, Vol. 105, (Jun-Jul. 1998), pp.565-566.

\bibitem{TiH}
Olav Tirkkonen and Ari Hottinen, ``Square-Matrix Embeddable Space-Time Block Codes for Complex Signal Constellations,'' IEEE Trans. Inf. Theory, vol. 48, no. 2, pp. 384-395, Feb. 2002.


\bibitem{EFE}
E. Bayer-Fluckiger, F. Oggier, E. Viterbo: "New Algebraic Constructions of Rotated $\mathbb{Z}^{n}$-Lattice Constellations for the Rayleigh Fading Channel,"\emph{ IEEE Trans. Inf. Theory}, vol. 50, no. 4, pp.702-714, April 2004.
\bibitem{FDRV}
Full Diversity Rotations [Online]. Available@ http://www1.tlc.polito.it/
$\sim$viterbo/rotations/rotations.html

%

%



%
%
\end{thebibliography}
\end{document}